\documentclass[peerreview, draftcls, onecolumn, 11pt]{IEEEtran}

\usepackage {graphicx}
\usepackage {amssymb}
\usepackage {amsmath}
\usepackage {mathrsfs}
\usepackage {algorithmic}
\usepackage {algorithm}

\usepackage{amsmath}

\begin{document}
%
\title{Efficient Implementation of Linear Programming Decoding}
\author{Mohammad H. Taghavi, Amin Shokrollahi, and Paul H. Siegel
\thanks{M. H. Taghavi and P. H. Siegel are with the Electrical and Computer Engineering Department, and the Center for Magnetic Recording Research, University of California, San Diego, La Jolla, CA 92093-0407 USA (e-mail: mtaghavi@ucsd.edu; psiegel@ucsd.edu).}
\thanks{A. Shokrollahi is with the School of Basic Sciences, and School of Computer Science and Communications, Ecole Polytechnique F\'ed\'erale de Lausanne (EPFL), 1015 Lausanne, Switzerland (e-mail: amin.shokrollahi@epfl.ch).}
\thanks{The material in this paper was presented in part at the 46th annual Allerton Conference on Communication, Control and Computing, Sep. 2008.}}


%
%
%

\maketitle

\begin{abstract}
Linear programming (LP) decoding, originally proposed by Feldman \emph{et al.} \cite{Feldman} as an approximation to the maximum-likelihood (ML) decoding of binary linear codes, solves a linear optimization problem formed by relaxing each of the finite-field parity-check constraints into a number of linear constraints. While providing a number of advantages over iterative message-passing (IMP) decoders, such as its amenability to finite-length performance analysis, LP decoding is computationally more complex to implement in its original form than IMP decoding, due to both the large size of the relaxed LP problem and the inefficiency of using general-purpose LP solvers. 

This paper explores ideas for fast LP decoding of low-density parity-check (LDPC) codes. We first show a number of properties of the LP decoder, and by modifying the previously reported Adaptive LP decoding scheme \cite{ALP ISIT} to allow removal of unnecessary constraints, we prove that LP 
decoding can be performed by solving a number of LP problems that contain at most one linear constraint derived from each of the parity-check constraints. Then, as a step toward designing an efficient LP solver that takes advantage of the particular structure of LDPC codes, we study a sparse interior-point method for solving this sequence of linear programs. Since the most complex part of each iteration of the interior-point algorithm is the 
solution of a (usually ill-conditioned) system of linear equations for finding the step direction, we propose a preconditioning algorithm to be used with the 
preconditioned conjugate-gradient method for solving such systems. The proposed preconditioning algorithm is similar to the encoding procedure of LDPC codes, and we demonstrate its effectiveness via both analytical methods and computer simulation results.
\end{abstract}

\IEEEpeerreviewmaketitle

\newtheorem{definition}{Definition}
\newtheorem{theorem}{Theorem}
\newtheorem{lemma}{Lemma}
\newtheorem{claim}{Claim}
\newtheorem{corollary}{Corollary}
\newtheorem{conjecture}{Conjecture}
\newtheorem{remark}{Remark}

\long\def\symbolfootnote[#1]#2{\begingroup%
\def\thefootnote{\fnsymbol{footnote}}\footnote[#1]{#2}\endgroup}

\section{Introduction}
Low-density parity-check (LDPC) codes \cite{Gallager} are becoming one of the dominant means of error-control coding in the transmission and storage of digital information. By combining randomness and sparsity, LDPC codes with large block lengths can correct errors using iterative message-passing (IMP) algorithms at coding rates that are closer to the capacity than any other class of practical codes \cite{RSU}. While the performance of IMP decoders for the asymptotic case of infinite lengths is studied extensively using probabilistic methods such as density evolution \cite{Capacity of LDPC}, the finite-length behavior of these algorithms, especially their error floors, are still not well-characterized. 

Linear programming (LP) decoding was proposed by Feldman \emph{et al.} \cite{Feldman} as an alternative to IMP decoding of LDPC and turbo-like codes. LP decoding approximates the maximum-likelihood (ML) decoding problem by a linear optimization problem via a relaxation of each of the finite-field parity-check constraints of the ML decoding into a number of linear constraints. Many observations suggest similarities between the performance of LP and iterative message-passing decoding methods \cite{Feldman}, \cite{LP and MSA}, \cite{Vontobel}. In fact,
the sum-product message-passing algorithm can be interpreted as a minimization of a nonlinear function, known as Bethe free energy, over the same feasible region as LP decoding \cite{Yedidia}, \cite{tree-reweighted}.

Due to its geometric structure, LP decoding seems to be more amenable than IMP decoding to finite-length analysis. In particular, the finite-length behavior of LP decoding can be completely characterized in terms of pseudocodewords, which are the vertices of the feasible space of the corresponding linear program. Another characteristic of LP decoding -- the \emph{ML certificate property} -- is that its failure to find an ML codeword is always detectable. More specifically, the decoder always gives either an ML codeword or a nonintegral pseudocodeword as the solution. On the other hand, the main disadvantage of LP decoding is its higher complexity compared to IMP decoding.

In order to make linear programming (LP) decoding practical, it is necessary to find efficient implementations that make its time complexity comparable to those of the message-passing algorithms. A conventional implementation of LP decoding is highly complex due to two main factors: (1) the large size of the LP problem formed by relaxation, and (2) the inability of general-purpose LP solvers to solve the LP efficiently by taking advantage of the properties of the decoding problem.

The standard formulation of LP decoding \cite{Feldman} has a size that grows very rapidly with the density of the Tanner graph representation of the code. Adaptive LP (ALP) decoding was proposed in \cite{ALP ISIT} to address this problem, reducing LP decoding to solving a sequence of much smaller LP problems.
The size of these LP problems has been observed in practice to be independent of the degree distribution, and more specifically, less than a small factor (less than two) times the number of parity checks. However, this observation has not been analytically explained. 

More recently, an equivalent formulation of the LP decoding problem was proposed in \cite{Chertkov} and \cite{Yang}, with a problem size growing linearly with both the code length and the maximum check node degrees. While this formulation requires solving only one LP, the overall complexity of this method in practice remains substantially higher than that of ALP decoding.

In this paper, we take some steps toward designing efficient LP solvers for LP decoding that exploit the inherent sparsity and structure of this particular class of problems. Our approach is based on a sparse implementation of interior-point algorithms. In an independent work, Vontobel studied the implementation and convergence of interior-point methods for LP decoding and mentioned a number of potential approaches to reduce its complexity \cite{Vontobel IPM}. It is also worth noting that a different line of work in this direction has been to apply iterative methods based on message-passing, instead of general LP solvers, to perform the optimization for LP decoding; e.g. see \cite{tree-reweighted} and \cite{Towards low-complexity}.

We first propose two modified versions of ALP decoding. The main idea behind these modifications is to adaptively remove a number of constraints at each iteration of ALP decoding, while adding new constraints to the problem. We prove a number of properties of these algorithms, which facilitate the design of a low-complexity LP solver. In particular, we show that the modified ALP decoders have the \emph{single-constraint property}, which means that they perform LP decoding by solving a series of linear programs that each contain at most one linear constraint from each parity check. 
An important consequence of this property is that the constraint matrices of the linear programs that are solved have a structure similar, in terms of the locations of their nonzero entries, to that of the parity-check matrix.

Then, we focus on the most complex part of each iteration of the interior-point algorithm, which is solving a system of linear equations to compute the Newton step. Since these linear systems become ill-conditioned as the interior-point algorithm approaches the solution, iterative methods that are often used for solving sparse systems, such as the conjugate-gradient (CG) method, perform poorly in the later iterations of the optimization. To address this problem, we propose a criterion for designing preconditioners that take advantage of the properties of LP decoding, along with a number of greedy algorithms to search for such preconditioners. The proposed preconditioning algorithms have similarities to the encoding procedure of LDPC codes, and we demonstrate their effectiveness via both analytical methods and computer simulation results.

The rest of this paper is organized as follows. In Section II, we review codes, LP decoding, and ALP decoding. In Section III, we propose some modifications in ALP decoding, and demonstrate a number of properties of ALP decoding and its variations. In Section IV, we review a class of the interior-point linear programming methods, as well as the preconditioned conjugate gradient (PCG) method for solving linear systems, with an emphasis on sparse implementation. In Section V, we introduce the proposed preconditioning algorithms to improve the PCG method for LP decoding. Some theoretical analysis and computer simulation results are presented in Section VI, and some concluding remarks are given in Section VII.

\section{LP Decoding}
\subsection{Notation}
Throughout the paper, we denote scalars and column vectors by lower-case letters ($a$), matrices by upper-case letters ($A$), and sets by calligraphic upper-case letters ($\mathcal{A}$). We write the $i$th element of a vector $a$ and the $(i,j)$th element of a matrix $A$ as $a_i$ and $A_{i,j}$, respectively.
The cardinality (size) of a finite set $\mathcal{A}$ is shown by $|\mathcal{A}|$. The support set (or briefly, support) of a vector $a$ of length $n$ is the set of locations $i\in\{1, \ldots, n\}$ such that $a_i \neq 0$. Similarly, the fractional support of a vector $a \in \mathbb{R}^n$ is the set of locations $i\in\{1, \ldots, n\}$ such that $a_i \notin \mathbb{Z}$.

A binary linear code $\mathcal{C}$ of block length $n$ is a subspace of $\{0, 1\}^n$. This supspace can be defined as the null space (kernel) of a parity-check matrix $H \in \{0 , 1\}^{m \times n}$ in modulo-2 arithmetic. In other words,
\begin{equation}
\label{linear code}
\mathcal{C} = \big\{u \in \{0, 1\}^n \big| H x = 0 \mod 2\big\}.
\end{equation}
Hence, each row of $H$ corresponds to a binary parity-check constraint. The design rate of this code is defined as $R = 1- \frac{m}{n}$. In this paper, we assume that $H$ has full row rank (mod 2), in which case the design rate is the same as the rate of the code.

Given the $m \times n$ parity-check matrix, $H$, the code can also be described by a Tanner graph. The Tanner graph $\mathcal{T}$ is a bipartite graph containing $n$ \emph{variable nodes} (corresponding to the columns of $H$) and $m$ \emph{check nodes} (corresponding to the rows of $H$). We denote by $\mathcal{I} = \{1, \ldots, n\}$ the set of (indices of) variable nodes, and by $\mathcal{J} = \{1, \ldots, m\}$ the set of (indices of) check nodes. Variable node $i$ is connected to check node $j$ via an edge in the Tanner graph if $H_{j, i} = 1$. 

The neighborhood $\mathcal{N}(j)$ of a check (variable) node $j$ is the set of variable (check) nodes it is directly connected to via an edge, i.e., the support set of the $j$th row (column) of $H$. The degree $d_j$ of a node $j$, where the type of the node will be clear from the context, is the cardinality of its neighborhood. 
Let $\mathcal{S} \subseteq \mathcal{I}$ be a subset of the variable nodes. We call $\mathcal{S}$ a \emph{stopping set} if there is no check node in the graph that has exactly one neighbor in $\mathcal{S}$. Stopping sets characterize the termination of a belief propagation erasure decoder.

Each code can be equivalently represented by many different parity-check matrices and Tanner graphs. However, it is important to note that the performance of suboptimal decoders, such as message-passing or LP decoding, may depend on the particular choice of $H$ and $\mathcal{T}$. A low-density parity-check (LDPC) code is a linear code which has at least one sparse Tanner graph representation, where the average variable node and check node degrees do not grow with $n$ or $m$.

A linear program (LP)\footnote{Throughout the paper, we abbreviate the terms ``linear program'' and ``linear programming'' both as ``LP''.} of dimension $n$ is an optimization problem with a linear objective function and a feasible set (space) described by a number of linear constraints (inequalities or equations) in terms of $n$ real-valued variables. Each linear constraint in the LP defines a hyperplane in $n$-dimensional space. If the solution to an LP is bounded and unique, then it is at a vertex $v$ of the feasible space, on the intersection of at least $n$ such hyperplanes. Conversely, for any vertex $v$ of the feasible space of an LP, there exists a choice of the coefficients of the objective function such that $v$ is the unique solution to the LP.
%
\subsection{LP Relaxation of Maximum-Likelihood Decoding}
Consider a binary linear code $\mathcal{C}$ of length $n$. If a codeword $v \in \mathcal{C}$ is transmitted through a memoryless binary-input output-symmetric (MBIOS) channel, the ML codeword $u^{ML}$ given the received vector $ r \in \mathbb{R}^n$ is the codeword that maximizes the likelihood of observing $r$, i.e., 
\begin{equation}
\label{Maximum Likelihood}
u^{ML} = \mathop{\arg\max}\limits_{u \in \mathcal{C}} \Pr[r | u].
\end{equation}
For binary codes, this problem can be rewritten as the equivalent optimization problem

\begin{equation}
\label{ML decoding}
\textbf{ML Decoding} \hspace{0.2in}
\setlength{\nulldelimiterspace}{0pt}
\left.\begin{IEEEeqnarraybox}[\relax][c]{l's}
\text{minimize} \hspace{0.2in} \gamma^T  u \\
\text{subject to} \hspace{0.18in} u\in \mathcal{C},
\end{IEEEeqnarraybox}\right. \vspace{0.1in}
\end{equation}
where $\gamma$ is the vector of log-likelihood ratios (LLR) defined as
\begin{equation}
\label{def gamma}
\gamma_i= \log \frac{\Pr(r_i|u_i=0)}{\Pr(r_i|u_i=1)}.
\end{equation}

The ML decoding problem (\ref{ML decoding}) is an optimization with a linear objective function in the real domain, but with constraints that are nonlinear in the real space (although, linear in modulo-2 arithmetic). It is desirable to replace these constraints by a number of linear constraints, such that decoding can be performed using linear programming. The feasible space of the desired LP would be the convex hull of all the codewords in $\mathcal{C}$, which is called \emph{the codeword polytope}. Since a global minimum occurs at one of the vertices of the polytope, using this feasible space makes the set of potential (unique) solutions to the LP identical to the set of codewords in $\mathcal{C}$. 
Unfortunately, the number of constraints needed for this LP representation grows exponentially with the code length, therefore making this approach impractical.
As an approximation to ML decoding, Feldman \emph{et al.} proposed a relaxed version of this problem by first considering the convex hull of the local codewords defined by each row of the parity-check matrix, and then intersecting them to obtain what is known as the \emph{fundamental polytope}, $\mathcal{P}$ \cite{Vontobel}.

To describe the (projected) fundamental polytope, linear constraints are derived from a parity-check matrix as follows. For each row $j=1,\ldots,m$ of the parity-check matrix, i.e., each check node, the LP formulation includes the constraints
\begin{equation}
\label{constraints}
\sum_{i\in \mathcal{V}} u_i -\sum_{i\in \mathcal{N}(j)\backslash \mathcal{V}} u_i \leq |\mathcal{V}|-1,\ \ \forall\ V\subseteq \mathcal{N}(j)\ 
\textrm{such that}\ |\mathcal{V}|\ \textrm{is odd},
\end{equation}
which can be written in the equivalent form
\begin{equation}
\label{constraints2}
\sum_{i\in \mathcal{V}} (1-u_i) + \sum_{i\in \mathcal{N}(j)\backslash \mathcal{V}} u_i \geq 1,\ \ \forall\ \mathcal{V}\subseteq \mathcal{N}(j)\ 
\textrm{such that}\ |\mathcal{V}|\ \textrm{is odd}.
\end{equation}
We refer to the constraints of this form as \emph{parity inequalities}. If the variables $u_i$ are zeroes and ones, these constraints will be equivalent to the original binary parity-check constraints. To see this, note that if $\mathcal{V}$ is a subset of $\mathcal{N}(j)$, with $|\mathcal{V}|$ odd, and the corresponding parity inequality fails to hold, then all variable nodes in $\mathcal{V}$ must have the value 1, while those in $\mathcal{N}(j)\backslash \mathcal{V}$ must have the value 0. This implies that the corresponding vector $u$ does not satisfy parity check $j$. 
Conversely, if parity check $j$ fails to hold, there must be a subset of variable nodes $\mathcal{V}\subseteq \mathcal{N}(j)$ of odd size such that all nodes in $\mathcal{V}$ have the value 1 and all those in $\mathcal{N}(j)\backslash \mathcal{V}$ have the value 0. Clearly, the  corresponding parity inequality would be violated. Now, given this equivalence, we relax the LP problem by replacing each binary constraint, $u_i \in \{0,1\}$, by a \emph{box constraint}, $0\leq u_i \leq 1$. LP decoding can then be written as

\begin{equation}
\label{LP decoding}
\textbf{LP Decoding} \hspace{0.2in}
\setlength{\nulldelimiterspace}{0pt}
\left.\begin{IEEEeqnarraybox}[\relax][c]{l's}
\text{minimize} \hspace{0.2in} \gamma^T  u \\
\text{subject to} \hspace{0.18in} u\in \mathcal{P}.
\end{IEEEeqnarraybox}\right. \vspace{0.1in}
\end{equation}

\begin{lemma}[\cite{Feldman}, originally by \cite{Jeroslow}]
\label{convex hull}
For any check node $j$, the set of parity inequalities (\ref{constraints}) defines the convex hull of all $0-1$ assignments of the variables with indices in $\mathcal{N}(j)$ that satisfy the $j$th binary parity-check constraint.
\end{lemma}

Since the convex hull of a set of vectors in $[0 , 1]^k$ is a subset of $[0 , 1]^k$, the set of parity inequalities for each check node automatically restrict all the involved variables to the interval $[0,1]$. Hence, we obtain the following corollary:
\begin{corollary}
\label{redundant box}
In the formulation of LP decoding above, the box constraints for variables that are involved in at least one parity-check constraint are redundant.
\end{corollary}

The fundamental polytope has a number of integral (binary-valued) and nonintegral (fractional-valued) vertices. The integral vertices, which satisfy all the parity-check equations as shown before, exactly correspond to the codewords of $\mathcal{C}$. Therefore, the LP relaxation has the \emph{ML certificate property}, i.e., whenever LP decoding gives an integral solution, it is guaranteed to be an ML codeword. On the other hand, if LP decoding gives as the solution one of the nonintegral vertices, which are known as \emph{pseudocodewords}, the decoder declares a failure.

\subsection{Adaptive Linear Programming Decoding}
In the original formulation of Feldman \emph{et al.} for LP decoding, the number of parity inequalities for each check node of degree $d_j$ is equal to the number of odd-sized subsets of its neighborhood, which is equal to $2^{d_j-1}$. Even for parity-check matrices of moderate row weights, this number can be very large. In \cite{ALP ISIT} a cutting-plane algorithm was proposed as an alternative to the direct implementation of LP decoding (\ref{LP decoding}). In this method, referred to as ``adaptive LP decoding'' (ALP decoding), a hierarchy of linear programs with the same objective function as in (\ref{LP decoding}) are solved, with the solution to the last program being identical to that of LP decoding. The first linear program in this hierarchy is made up of only $n$ box constraints, such that for each $i\in \{1, 2, \ldots, n\}$, we include the constraint
\begin{equation}
\label{initial constraints}
\setlength{\nulldelimiterspace}{0pt}
\left\{\begin{IEEEeqnarraybox}[\relax][c]{l's}
0\leq u_i \hspace{0.2in} \text{if}\ \ \gamma_i\geq0, \\
u_i\leq 1 \hspace{0.2in} \text{if}\ \ \gamma_i<0.
\end{IEEEeqnarraybox}\right.
\end{equation}
The solution to this initial problem corresponds to the result of an (uncoded) bit-wise hard decision based on the received vector.

\begin{algorithm}
\caption{ALP Decoding}
\label{adaptive LP}
\begin{algorithmic}[1]
\STATE Setup the initial LP problem with constraints from (\ref{initial constraints}), and $k \leftarrow 0$;
\STATE Find the solution $u^0$ to the initial LP problem by bit-wise hard decision;
\REPEAT
\STATE $k \leftarrow k+1$;
\STATE Find the set $\mathcal{S}^k$ of all parity inequalities and box constraints that are violated at $u^{k-1}$;
\STATE If $|\mathcal{S}^k|>0$, add the constraints in $\mathcal{S}^k$ to the LP problem and solve it to obtain $u^k$;
\UNTIL{$|\mathcal{S}^k|=0$}
\STATE Output $u = u^k$ as the solution to LP decoding.
\end{algorithmic}
\end{algorithm}

The adaptive LP decoding algorithm is presented here as Algorithm \ref{adaptive LP} (ALP decoding). In Step 5 of this algorithm, the search for all the violated parity inequalities can be performed using Algorithm 1 of \cite{ALP ISIT} in $O(\sum_{i=1}^m d_j \log d_j) = O(m d_{max} \log d_{max})$ time, without having to examine all the $O(m 2^{d_{max}})$ parity inequalities given by the original LP decoding formulation. Furthermore, based on observations, it is conjectured in \cite{ALP Journal} that there is no need to check for violated box constraints in Step 5, since they cannot be violated at any of the intermediate solutions $u^k$ of ALP decoding. In the next section, we present a proof of this conjecture.

In \cite{ALP ISIT}, the number of iterations of ALP decoding was upper-bounded by the code length, $n$. However, it was observed in the simulations that the typical number of iterations is much smaller in practice (less than $20$ for all $n < 2000$). Moreover, one can conclude from the following theorem that, at each iteration of ALP decoding, the number of violated parity inequalities added to the problem is at most $m$, where $m$ is the number of check nodes.
\begin{theorem}[\cite{ALP Journal}]
\label{one cut}
If at any given point $u \in [0,1]^n$, one of the parity inequalities introduced by a check node $j$ is violated, the rest of the parity inequalities from this check node are satisfied with strict inequality.
\end{theorem}
%
%
%
%
%
%
%
\section{Properties and Variations of ALP decoding}
In this section, we prove some properties of LP and ALP decoding, and propose some modifications to the ALP algorithm. As we will see, many of the elegant properties of these algorithms are consequences of Theorem \ref{one cut}.

First, we propose an alternative to using Algorithm 1 of \cite{ALP ISIT} for finding all the violated parity inequalities at any given point $u \in [0 , 1]^n$. Consider the general form of parity inequalities in (\ref{constraints2}) for a given check node $j$, and note that at most one of these inequalities can be violated at $u$. To find this inequality, if it exists, we need to find an odd-sized $\mathcal{V}\subseteq \mathcal{N}(j)$ that minimizes the left-hand side of (\ref{constraints2}). If there were no requirement that $|\mathcal{V}|$ is odd, the left-hand side expression would be minimized by putting any $i \in \mathcal{N}(j)$ with $u_i \geq \frac{1}{2}$ in $\mathcal{V}$. However, if such $\mathcal{V}$ has an even cardinality, we need to select one element $i^*$ of $\mathcal{N}(j)$ to add to or remove from $\mathcal{V}$, such that the increase on the left-hand side of (\ref{constraints2}) is minimal. This means that $i^*$ is the element whose corresponding value $u_{i^*}$ is closest to $\frac{1}{2}$. This results in Algorithm \ref{find parity cuts}, which has $O(d_j)$ complexity for check node $j$, thus reducing the complexity of finding all the $m$ parity inequalities from $O(\sum_{i=1}^m d_j \log d_j)$
with Algorithm 1 of \cite{ALP ISIT} to $O(\sum_{j=1}^m d_j) = O(E)$, where $E$ is the total number of edges in the Tanner graph.

\begin{algorithm}
\caption{Find the Violated Parity Inequality from Check Node $j$ at $u$}
\label{find parity cuts}
\begin{algorithmic}[1]
\STATE $\mathcal{S} \leftarrow \{i \in \mathcal{N}(j) | u_i>\frac{1}{2}\}$;
\IF {$|\mathcal{S}|$ is odd}
\STATE $\mathcal{V}\leftarrow \mathcal{S}$;
\ELSE 
\STATE $i^* \leftarrow \arg\min_{i\in \mathcal{N}(j)} |u_i-\frac{1}{2}|$;
\STATE $\mathcal{V}\leftarrow \mathcal{S}\backslash \{i^*\}$ if $i^* \in \mathcal{S}$; otherwise $\mathcal{V}\leftarrow \mathcal{S}\cup \{i^*\}$;
\ENDIF
\IF{ (\ref{constraints2}) is satisfied at $u$ for this $j$ and $\mathcal{V}$}
\STATE Check node $j$ does not introduce a violated parity inequality at $u$;
\ELSE \STATE We have found the violated parity inequality from check node $j$;
\ENDIF
\end{algorithmic}
\end{algorithm}

\subsection{Modified ALP Decoding}
\begin{definition}
A linear inequality constraint of the form $a^T x\leq b$ is called \emph{active} at point $x^0$ if it holds with equality; i.e., $a^T x^0 = b$, and is called \emph{inactive} if it holds with strict inequality; i.e. $a^T x^0 < b$.
\end{definition}
The following is a corollary of Theorem \ref{one cut}
\begin{corollary}
\label{one cut cor}
If one of the parity inequalities introduced by a check node is active at a point $x^0 \in [0 , 1]^n$, all parity inequalities from this check node must be satisfied at $x^0$.
\end{corollary}

Corollary \ref{one cut cor} can be used to simplify Step 5 of ALP decoding (Algorithm \ref{adaptive LP}) as follows. We first find the parity inequalities currently in the problem that are active at the current solution, $u^k$. This can be done simply by checking if the slack variable corresponding to a constraint is zero. Then, in the search for violated constraints, we exclude the check nodes that introduce these active inequalities.

Now consider the linear program $LP^k$ at an iteration $k$ of ALP decoding, with an optimum point $u^k$. This point is the vertex (apex) of the $n$-dimensional cone formed by all hyperplanes corresponding to the active constraints. It is easy to see that among the constraints in this linear program, the inactive ones are \emph{non-binding}, meaning that, if we remove the inactive constraints from the problem, $u^k$ remains an optimum point of the feasible space. This fact motivates a modification in the ALP decoding algorithm, where, after solving each LP, a subset of the constraints that are active at the solution are removed. 

By combining the two ideas proposed above, we obtain the modified ALP decoding algorithm A (MALP-A decoding), stated in Algorithm \ref{Modified ALP}. 
It was conjectured in \cite{ALP Journal} that no box constraint can be violated at any intermediate solution of ALP decoding. We will prove this conjecture for both ALP and MALP decoding in this section. Hence, we do not search for violated box constraints in the intermediate iterations of the proposed algorithms.

\begin{algorithm}
\caption{MALP-A Decoding}
\label{Modified ALP}
\begin{algorithmic}[1]
\STATE Setup the initial LP problem with constraints from (\ref{initial constraints}), and $k \leftarrow 0$;
\STATE Find the solution $u^0$ to the initial LP problem by bit-wise hard decision;
\REPEAT
\STATE $k \leftarrow k+1;\ flag \leftarrow 0;$
\FOR{$j=1$ to $m$}
\IF{there is no active parity inequality from check node $j$ in the problem}
\IF{check node $j$ introduces a parity inequality that is violated at $u^{k-1}$}
\STATE Remove the parity inequalities of check node $j$ (if any) from the current problem; 
\STATE Add the new (violated) constraint to the LP problem; $flag \leftarrow 1$;
\ENDIF
\ENDIF
\ENDFOR
\STATE If $flag=1$, solve the LP problem to obtain $u^k$;
\UNTIL{$flag=0$}
\STATE Output $u = u^k$ as the solution to LP decoding.
\end{algorithmic}
\end{algorithm}

Checking the condition in line 7 can be done using Algorithm \ref{find parity cuts} in $O(d_j)$ time, where $d_j$ is the degree of check node $j$, and the role of the if-then structure of line 6 is to limit this processing to only check nodes that are not currently represented in the problem by an active constraint. In line 8, before adding a new constraint from check node $j$ to the problem, any existing (inactive) constraint is removed from the problem. Alternatively, we can move this command to line 6; i.e. remove all the inactive constraints in the problem. We call the resulting algorithm the modified ALP decoding algorithm B (MALP-B decoding). 

The LP problems solved in the ALP and modified ALP decoding algorithms can be written in the ``standard'' matrix form as

\begin{equation}
\label{MALP matrix form}
\hspace{0.2in}
\setlength{\nulldelimiterspace}{0pt}
\left.\begin{IEEEeqnarraybox}[\relax][c]{l's}
\text{minimize} \hspace{0.2in} \gamma^T  u \\
\text{subject to} \hspace{0.18in} A u \leq b,\\
\hspace{0.85in} u_i\geq 0 \hspace{0.2in} \forall\ i\in \mathcal{I}:\ \gamma_i\geq0, \\
\hspace{0.85in} u_i\leq 1 \hspace{0.2in} \forall\ i\in \mathcal{I}:\ \gamma_i<0,
\end{IEEEeqnarraybox}\right. \vspace{0.1in}
\end{equation}
where matrix $A$ is called the \emph{constraint matrix}.
\subsection{Properties}
In Theorem 2 of \cite{ALP ISIT}, it has been shown that the sequence of solutions to the intermediate LP problems in ALP decoding converges to that of LP decoding in at most $n$ iterations. In the following theorem, in addition to proving that this property holds for the two modified ALP decoding algorithms, we show three additional properties shared by all three variations of adaptive LP decoding.

We assume that the optimum solutions to all the LP problems in the intermediate iterations of either ALP, MALP-A, or MALP-B decoding are unique. However, one can see that this uniqueness assumption is not very restrictive, since it holds with high probability if the channel output has a finite probability density function (pdf). Moreover, channels that do not satisfy this property, such as the binary symmetric channel (BSC), can be modified to do so by adding a very small continuous-domain noise to their output (or LLR vector).

\begin{theorem}[Properties of adaptive LP decoding]
\label{properties of ALP}
Let $u^0, u^1, \ldots, u^K$ be the unique solutions to the sequence of LP problems, $LP^0, LP^1, \ldots, LP^K$, solved in either ALP, MALP-A, or MALP-B decoding algorithms. Then, the following properties hold for all three algorithms:
\begin{enumerate}
\renewcommand{\labelenumi}{\alph{enumi})}
\item The sequence of solutions $u^0, u^1, \ldots$ satisfy all the box constraints $0\leq u_i \leq 1,\ \forall\ i=1, \ldots, n$. 
\item The costs of these solutions monotonically increase with the iteration number; i.e., 
\begin{equation}
\gamma^T u^0 < \gamma^T u^1 < \ldots
\end{equation}
\item $u^0, u^1, \ldots$ converge to the solution of LP decoding, $u^*$, in at most $n$ iterations.
\item Consider the set of parity inequalities included in $LP^k$ which are active at its optimum solution, $u^k$. Let $\mathcal{J}^k = \{j_1, j_2, \ldots, j_{|\mathcal{J}^k|}\}$ be the set of indices of check nodes that generate these inequalities.
Then, $u^k$ is the solution to an LP decoding problem $LPD^k$ with the LLR vector $\gamma$ and the Tanner graph corresponding to the check nodes in $\mathcal{J}^k$.
\end{enumerate} 
\end{theorem}
The proof of this theorem is given in Appendix \ref{Proof of theorem}.

The following theorem shows an interesting property of the modified ALP decoding schemes, which we call the ``single-constraint property.'' This property does not hold for ALP decoding.
\begin{theorem}
\label{one ineq per check}
In the LP problem at any iteration $k$ of the MALP-A and MALP-B decoding algorithms, there is at most one parity inequality corresponding to each check node of the Tanner graph.
\end{theorem}
\begin{proof}[By induction]
The initial LP problem consists only of box constraints. So, it suffices to show that, if the LP problem $LP^k$ at an iteration $k$ satisfies the desired property, the LP problem $LP^{k+1}$ in the subsequent iteration satisfies this property, as well. Consider check node $j$ which has a violated parity inequality $\kappa_j$ at the solution $u^k$ of $LP^k$. According to Corollary \ref{one cut cor}, if there already has been a parity inequality $\tilde\kappa_j$ from this check node in $LP^k$, $\tilde\kappa_j$ cannot be active at $u^k$, hence, the MALP decoder will remove $\tilde\kappa_j$ before adding $\kappa_j$ to $LP^{k+1}$. As a result, there cannot be more than one parity inequality from any check node $j$ in $LP^{k+1}$
\end{proof}

\begin{corollary}
\label{m constraints}
The number of parity inequalities in any linear program solved by the MALP decoding algorithms is at most $m$
\end{corollary}

The result above is in contrast to the non-adaptive formulations of LP decoding, where the size of the LP problems grows with the check node degree. Consequently, the complexity of these two algorithms can be bounded by their number of iterations times the worst-case complexity of solving an LP problem with $n$ variables and $m$ parity inequalities. Therefore, an interesting problem to investigate is how the number of iterations of the MALP decoding algorithms varies with the code parameters, such as the length or the check node degrees, and how its behavior changes depending on whether the LP decoding output is integral or fractional. In Subsection III-D, we present some simulation results, studying and comparing ALP decoding and its modifications in terms of the number of iterations.

An important consequence of Theorem \ref{one ineq per check} is that, in the LP problems that are solved by these two algorithms, the distribution of the nonzero elements of the LP constraint matrix, $A$, has the same structure as that of the parity-check matrix, $H$, after removing the rows of $H$ that are not represented by a parity inequality in the LP. This is due to the fact that the support set of a row of $A$, corresponding to a parity inequality, is identical to that of the row of $H$ from which it has been derived, and in addition, each row of $A$ is derived from a unique row of $H$. 
As we will see later in this paper, this property, which is not shared by LP or ALP decoding, maintains the same desirable combinatorial properties (e.g., degree distribution) for $A$ that the $H$ matrix has. This can be exploited in the design of efficient LP solvers.

Remember that the LP problem in the last iteration of the MALP decoding algorithms has the same solution as standard LP decoding. This solution is a vertex of the feasible set, defined by at least $n$ active inequalities from this LP problem. Hence, using Corollary \ref{m constraints}, we conclude that at least $n-m$ box constraints are active at the solution of LP decoding. This yields the following properties of LP decoding. 

\begin{corollary}
\label{n-m correct}
The solution to any LP decoding problem differs in at most $n-m$ coordinates from the vector obtained by making bit-based hard decisions on the LLR vector $\gamma$.
\end{corollary}

\begin{corollary}
\label{m fractional}
Each pseudocodeword of LP decoding has at most $m$ fractional entries.
\end{corollary}

\begin{remark}
This bound on the size of the fractional support of pseudocodewords is tight in the sense that there are LP decoding relaxations which have pseudocodewords with exactly $m$ fractional entries. An example is the pseudocodeword $[1, \frac{1}{2}, 0, \frac{1}{2}, 0, 0, \frac{1}{2}]$ of the $(7, 4, 3)$ code with $m=3$, given in \cite{Feldman}.
\end{remark}

\subsection{Connection to Erasure Decoding}
For the binary erasure channel (BEC), the performance of belief propagation (BP), or its equivalent, the peeling algorithm, has been extensively studied. The peeling algorithm can be seen as performing row and column permutations to triangularize a submatrix of $H$ consisting of the columns corresponding to the erased bits. It is known that the BP and peeling decoders succeed on the BEC if and only if the set of erased bits does not contain a stopping set. 

Feldman \emph{et al.} have shown in \cite{Feldman} that LP decoding and BP decoding are equivalent on the BEC. In other words, the success or failure of LP decoding can also be explained by stopping sets. In this subsection, we show a connection between LP decoding on the BEC and LP decoding on general MBIOS channels, allowing us to derive a sufficient condition for the failure of LP decoding on general MBIOS channels based on the existence of stopping sets.

\begin{theorem}
\label{connection to erasure}
Consider an LP decoding problem $LPD^0$ with LLR vector $\gamma$, $\gamma_i \neq 0\ \forall\ i\in\mathcal{I}$, resulting in the unique integral solution (i.e., the ML codeword) $u$. Also, let $\tilde{u}$ be the result of bit-based hard decisions on $\gamma$; i.e., $\tilde{u}_i=0$ if $\gamma_i>0$, and $\tilde{u}_i=1$ otherwise. Then, the set $\mathcal{E} \subseteq \mathcal{I}$ of positions where $u$ and $\tilde{u}$ differ, does not contain a stopping set.
\end{theorem}
\begin{proof}
Let's assume, without loss of generality, that $u$ is the vector of all-zeroes, in which case we will have
\begin{equation}
\label{E based on gamma}
\mathcal{E}= \big\{i\in \mathcal{I} | \gamma_i < 0\big\}. 
\end{equation}
We form an LP erasure decoding problem $LPD^{BEC}$ with $u$ as the transmitted codeword and $\mathcal{E}$ as the set of erased positions. $LPD^{BEC}$ has the same feasible space $\mathcal{P}$ as $LPD^0$, but has a new LLR vector $\lambda$, defined such that $\forall\ i \in \mathcal{I}$,
\begin{equation}
\label{lambda definition}
\setlength{\nulldelimiterspace}{0pt}
\lambda_i=
\left\{\begin{array}[\relax]{ll}
0 & \text{if}\ \ i \in \mathcal{E}, \\
1 & \text{otherwise},
\end{array}\right. 
\end{equation}
Clearly, since $\mathcal{P}\subseteq [0,1]^n$, we have $\lambda^T v \geq 0, \ \forall\ v\in\mathcal{P}$. We prove the theorem by showing that the all-zeroes vector $u$ is the unique solution to $LPD^{BEC}$, as well. 

Assume that there is another vector $v \in \mathcal{P}$ such that we have 
\begin{equation}
\label{BEC, cost of v}
\lambda^T v = \lambda^T u = 0.
\end{equation}
Combining (\ref{lambda definition}) and (\ref{BEC, cost of v}) yields
\begin{align}
\sum\limits_{i\in \mathcal{I}\backslash \mathcal{E}} v_i =0,
\end{align}
implying that $v_i=0, \ \forall\ i\in \mathcal{I}\backslash \mathcal{E}$.
Therefore, using (\ref{E based on gamma}), the cost of the vector $v$ for $LPD^0$ will be
\begin{align}
\gamma^T v &= \sum\limits_{i\in \mathcal{E}} \gamma_i v_i \nonumber\\
&\leq 0 = \gamma^T u,
\end{align}
with equality if and only if $v_i=0, \ \forall\ i\in\mathcal{I}$. Since, by assumption, $u$ is the unique solution to $LPD^0$, we must have $v=u =[0,\ldots, 0]^T$. Hence, $u$ is also the unique solution to $LPD^{BEC}$. Finally, due to the equivalence of LP and BP decodings on the BEC, we conclude that $\mathcal{E}$ does not contain a stopping set.
\end{proof}

Theorem \ref{connection to erasure} will be used later in the paper to design an efficient way to solve the systems of linear equations we encounter in LP decoding.

\subsection{Simulation Results}
We present simulation results for ALP, MALP-A, and MALP-B decoding of random $(3,6)$-regular
LDPC codes, where the cycles of length four are removed from the Tanner graphs of the codes. The
simulations are performed in an AWGN channel with the SNR of $2$ dB (the threshold of belief-propagation decoding for the ensemble of $(3,6)$-regular codes is $1.11$ dB), and include 8 different lengths, with 1000 trials at each length.

In Fig. \ref{Hist_iter}, we have plotted the histograms of the number of iterations using the three algorithms for length $n=480$. The first column of histograms includes the results of all the decoding instances, while the second and third columns only include the decoding instances with integral and fractional outputs, respectively. From this figure, we can see that when the output is integral (second column), the three algorithms have a similar behavior, and they all converge in less that 15 iterations. On the other hand, when the output is fractional (third column), the typical numbers of iterations are 2-3 times higher for all algorithms, so that we observe two almost non-overlapping peaks in the histograms of the first columns.

\begin{figure}
\centering
\includegraphics[width=4.0 in] {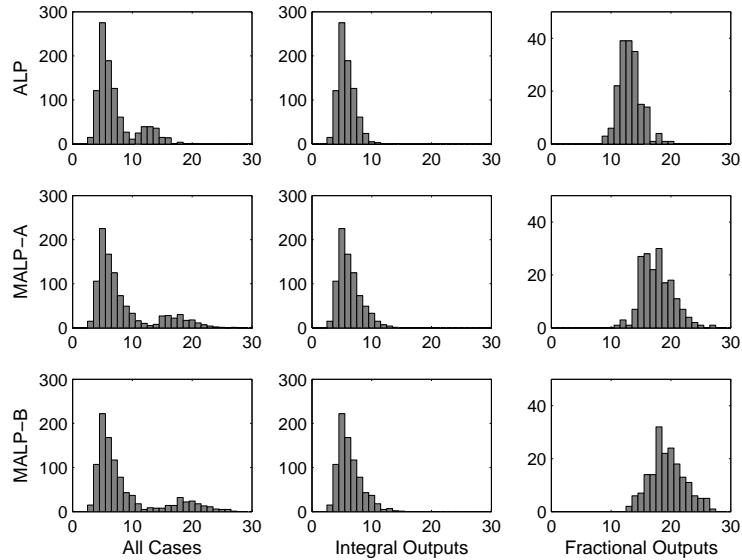}
\caption{The histograms of the number of iterations for ALP, MALP-A, and MALP-B decoding for a random $(3,6)$-regular LDPC code of length 480 at SNR = 2 dB. The left, middle, and right columns respectively correspond to the results of all decoding instances, decodings with integral outputs, and decodings with fractional output.}
\label{Hist_iter}
\end{figure}

In Fig. \ref{Iter_vs_n}, the average numbers of iterations of the three algorithms are plotted for both integral and fractional decoding outputs versus the code length. As a measure of the deviation of the results from the mean, we have also included the $95\%$ one-sided confidence upper bound for each curve, which is defined as the smallest number which is higher than at least $95\%$ of the values in the population. We can observe that the number of iterations for MALP-A and MALP-B decoding are significantly higher that that of ALP when the output is fractional. On the other hand, for decoding instances with integral outputs, where the LP decoder is successful in finding the ML codeword, the increase in the number of iterations for the modified ALP decoders relative to the ALP decoder is very small. Hence, the MALP decoders pay a small price in terms of the number of iterations in exchange for obtaining the single-constraint property. 

\begin{figure}
\centering
\includegraphics[width=4.3 in] {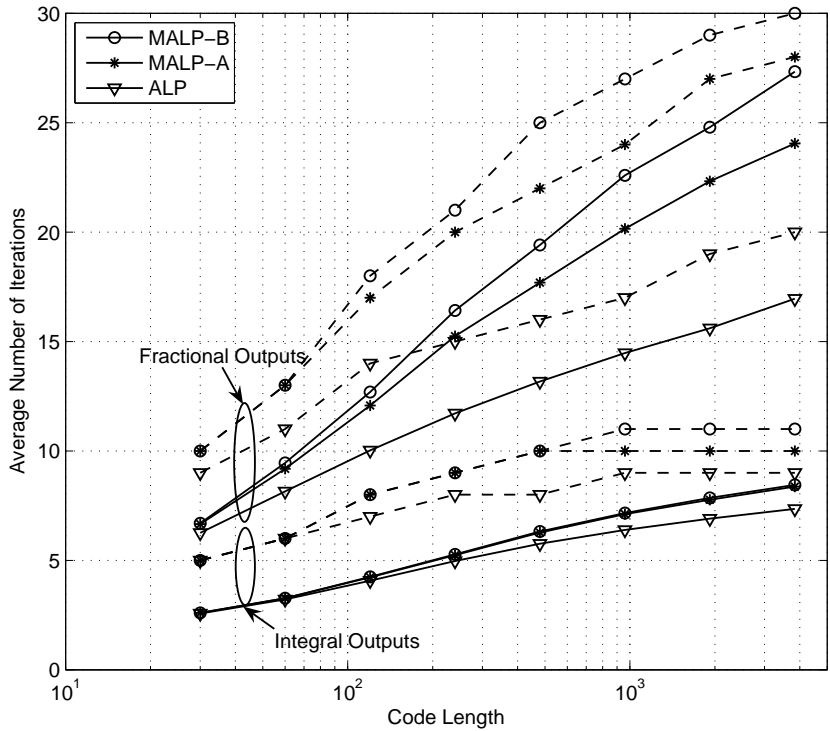}
\caption{The number of iterations of ALP, MALP-A, and MALP-B decoding versus code length for random $(3,6)$-regular LDPC codes at SNR = 2 dB. The solid and dashed curves represent, respectively, the average values and the $95\%$ one-sided confidence upper bounds.}
\label{Iter_vs_n}
\end{figure}

Moreover, our simulations indicate that the size of the largest LP that is solved in each MALP-A or MALP-B decoding problem is smaller on average than that of ALP decoding by $17\%$ for integral outputs and $30\%$ for fractional outputs.

\section{Solving the LP Using the Interior Point Method}
General-purpose LP solvers do not take advantage of the particular structure of the optimization problems arising in LP decoding, and, therefore, using them can be highly inefficient. In this and the next sections, we investigate how LP algorithms can be implemented efficiently for LP decoding. The two major techniques for linear optimization used in most applications are Dantzig's simplex algorithm \cite{Dantzig} and the interior point methods.

\subsection{Simplex vs. Interior-Point Algorithms}
The simplex algorithm takes advantage of the fact that the solution to an LP is at one of the vertices of the feasible polyhedron. Starting from a vertex of the feasible polyhedron, it moves in each iteration (pivot) to an adjacent vertex, until an optimal vertex is reached. Each iteration involves selecting an adjacent vertex with a lower cost, and computing the size of the step to take in order to move to that edge, and these are computed by a number of matrix and vector operations.

Intertior-point methods generally move along a path within the interior of the feasible region. Starting from an interior point, interior point methods approximate the feasible region in each iteration, and take a Newton-type step towards the next point, until they get to the optimum point. Computation of these steps involves solving a linear system.

The complexity of an LP solver is determined by the number of iterations it takes to converge and the average complexity of each iteration. The number of iterations of the simplex algorithm has been observed to be polynomial (superlinear), on average, in the problem dimension $n$, while its worst-case performance can be exponential. An intuitive way of understanding why the average number of simplex pivots to successfully solve an LP decoding problem is at least linear in $n$ is to note that each pivot makes one basic primal variable nonbasic (i.e. sets it to zero) and makes one nonbasic variable basic (i.e. possibly increases it from zero). Hence, starting from an initial point, it should generally take at least a constant times $n$ pivots to arrive at a point corresponding to a binary codeword. Therefore, even if the computation of each simplex iteration were done in linear time, one could not achieve a running time better that $O(n^2)$, unless the simplex method is fundamentally revised. 

In contrast to the simplex algorithm, for certain classes of iterior-point methods, such as the path-following algorithm, the worst-case number of iterations has been shown to be $O(\sqrt{n})$, although these algorithms typically converge in $O(\log n)$ iterations \cite{Bertsimas}. Therefore, if the Newton step at each iteration can be computed efficiently, taking advantage of the sparsity and structure in the problem, one could obtain an algorithm that is faster than the simplex algorithm for large-scale problems.

Interior-point methods consist of a variety algorithms, differing in the way the optimization problem is approximated by an unconstrained problem, and how the step is calculated at each iteration. One of the most successful classes of interior-point methods is the primal-dual path-following algorithm, which is most effective for large-scale applications. In the following subsection we present a brief review of this algorithm. For a more comprehensive description, we refer the reader to the literature on linear programming and interior-point methods.

\subsection{Primal-Dual Path-Following Algorithm}
For simplicity, in this section we assume that the LP problems that we want to solve are of the form (\ref{MALP matrix form}). However, by introducing a number of additional slack variables, we can modify all the expressions in a straighforward way to represent the case where both types of box constraints may be present for each variable.

We first write the LP problem with $q$ variables and $p$ constraints in the ``augmented'' form

\begin{equation}
\label{primal LP}
\textbf{Primal LP} \hspace{0.2in}
\setlength{\nulldelimiterspace}{0pt}
\left.\begin{IEEEeqnarraybox}[\relax][c]{l's}
\text{minimize} \hspace{0.2in} c^T x \\
\text{subject to} \hspace{0.18in} A x = b,\\
  \hspace{0.85in} x\geq 0.
\end{IEEEeqnarraybox}\right. \vspace{0.1in}
\end{equation}
Here, to convert the LP problem (\ref{MALP matrix form}) into the form above, we have taken two steps. First, noting that each variable $u_i$ in (\ref{MALP matrix form}) is subject to exactly one box constraint of the form $u_i\geq 0$ or $u_i\leq 1$, we introduce the variable vector $x$ and cost vector $c$, such that for any $i=1,\ldots,n$, $x_i=u_i$ and $c_i=\gamma_i$ if the former inequality is included (i.e., $\gamma_i\geq 0$), and $x_i=1-u_i$ and $c_i=-\gamma_i$, otherwise. Therefore, the box constraints will all have the form $x_i\geq 0$, and the coefficients of the parity inequalities will also change correspondingly. Second, for any $j=1,\ldots, p,$ we convert the parity inequality  $A_{j\diamond} x \leq b_j$ in (\ref{MALP matrix form}), where $A_{j\diamond}$ denotes the $j$th row of $A$, to a linear equation $A_{j\diamond} x + x_{n+j}= b_j,$ by introducing $p$ nonnegative slack variables $x_{n+1},\ldots,x_{q}$, where $q=n+p$, with corresponding coefficients equal to zero in the cost vector, $c$. We will sometimes refer to the first $n$ (non-slack) variables as the \emph{standard variables}. The dual of the primal LP has the form

\begin{equation}
\label{dual LP}
\textbf{Dual LP} \hspace{0.2in}
\setlength{\nulldelimiterspace}{0pt}
\left.\begin{IEEEeqnarraybox}[\relax][c]{l's}
\text{minimize} \hspace{0.2in} b^T y\\
\text{subject to} \hspace{0.18in} A^T y + z = c,\\
  \hspace{0.85in} z\geq 0,
\end{IEEEeqnarraybox}\right. \vspace{0.1in}
\end{equation}
where $y$ and $z$ are the dual standard and slack variables, respectively.

The first step in solving the primal and dual problems is to remove the inequality constraints by introducing logarithmic \emph{barrier terms} into their objective functions.\footnote{Because of this step, interior-point methods are sometime referred to in the literature as barrier methods.} The primal and dual objective functions will thus change to $c^T x - \mu \sum_{i=1}^q \log x_i$ and $b^T y - \mu \sum_{i=1}^q \log z_i$, respectively, for some $\mu>0$, resulting in a familiy of convex nonlinear barrier problems $P(\mu)$, parameterized by $\mu$, that approximate the original linear program. Since the logarithmic term forces $x$ and $z$ to remain positive, the solution 
to the barrier problem is feasible for the primal-dual LP, and it can be shown that as $\mu \to 0$,  it approaches the solution to the LP problem. The key idea of the path-following algorithm is to start with some $\mu>0$, and reduce it at each iteration, as we take one step to solve the barrier problem.

The Karush-Kuhn Tucker (KKT) conditions provide necessary and sufficient optimality conditions for $P(\mu)$, and can be written as \cite[Chapter 9]{Bertsimas}
\begin{align}
\label{KKT}
Ax &= b\\
A^T y + z &= c\\
X Z e &= \mu e \label{slackness}\\
x, z &\geq 0,
\end{align}
where $X$ and $S$ are diagonal matrices with the entries of $x$ and $z$ on their diagonal, respectively, and $e$ denotes the all-ones vector. If we define
\begin{equation*}
F(s)=\begin{bmatrix}
Ax - b\\
A^T y + z - c\\
X Z e - \mu e \\
\end{bmatrix},
\end{equation*}
where $s=(x, y, z)$ is the current primal-dual iterate, the problem of solving $P(\mu)$ reduces to finding the (unique) zero of the multivariate function $F(s)$. In Newton's method, $F(s)$ is iteratively approximated by its first order Taylor series expansion around  $s=s^k$
\begin{equation}
\label{Taylor}
F(s^k+\Delta s^k) \approx F(s^k) + J(s^k) \Delta s^k,
\end{equation}
where $J(s)$ is the Jacobian matrix of $F(s)$. The Newton direction $\Delta s^k = (\Delta x^k, \Delta y^k, \Delta z^k)$ is obtained by setting the right-hand side of (\ref{Taylor}) to zero, resulting in the following system of linear equations:
\begin{equation}
\label{matrix equation}
\begin{bmatrix}
A & 0 & 0\\
0 & A^T & I\\
Z_k & 0 & X_k
\end{bmatrix}
\begin{bmatrix}
\Delta x^k\\
\Delta y^k\\
\Delta z^k
\end{bmatrix}
=
\begin{bmatrix}
r_b\\
r_c\\
r_e
\end{bmatrix}
\end{equation}
where $r_b=b - Ax^k$, $r_c=c - A^T y^k - z^k$, and $r_e= \mu^k e - X_k Z_k e$ are the residuals of the KKT equations (\ref{KKT}), and $\mu^k$ is the value of $\mu$ at iteration $k$. 
If we start from a primal and dual feasible point, we will not need to compute $r_b$ and $r_c$, as they will remain zero throughout the algorithm. However, for sake of generality, here we do not make any feasibility assumption, in order to have the flexibility to apply the equations in the general, possibly infeasible case.

The solution to the linear system (\ref{matrix equation}) is given by
\begin{align}
(A D_k^2 A^T) \Delta y^k & = r_b + A D_k^2 r_c - A Z_k^{-1} r_e, \label{Delta y}\\
\Delta x^k & = D_k^2 A^T \Delta y^k - D_k^2 r_c  + Z_k^{-1} r_e, \label{Delta x}\\
\Delta z^k & = X_k^{-1} (r_e - Z \Delta x^k) \label{Delta z},
\end{align}
where 
\begin{equation}
\label{definition of D2}
D_k^2 = X_k Z_k^{-1}.
\end{equation}
To simplify the notation, we will henceforth drop the subscript $k$ from $D_k$, but it should be noted that $D$ is a function of the iteration number, $k$. Having the Newton direction, the solution is updated as
\begin{align*}
x^{k+1} &= x^k + \beta_P^k \Delta x^k,\\
y^{k+1} &= y^k + \beta_D^k \Delta y^k,\\
z^{k+1} &= z^k + \beta_D^k \Delta z^k,
\end{align*}
and the primal and dual step lengths, $\beta_P^k, \beta_D^k \in [0 , 1]$, are chosen such that all the entries of $x$ and $z$ remain nonnegative.

Since we are interested in solving the LP and not the barrier program $P(\mu)$ for a particular $\mu$, rather than taking many Newton steps to approach the solution to $P(\mu)$, we reduce the value of $\mu$ each time a Newton step is taken, so that barrier program gives a better approximation of the LP. A reasonable updating rule for $\mu$ is to make it proportional to the duality gap $g_d \triangleq (x^k)^T z^k$, that is
\begin{equation}
\label{update mu}
\mu^k = \frac{(x^k)^T z^k}{q}.
\end{equation}

The primal-dual path-following algorithm described above will iterate until the duality gap becomes sufficiently small; i.e. $(x^k)^T z^k < \epsilon$. It has been shown that with a proper choice of the step lengths, this algorithm takes $O\big( \sqrt{q} \log(\epsilon_0\slash \epsilon) \big)$ to reduce the duality gap from $\epsilon_0$ to $\epsilon$.

In order to initialize the algorithm, we need some feasible $x^0>0$, $y^0$, and $z^0>0$. Obtaining such an initial point is nontrivial, and is usually done by introducing a few dummy variables, as well as a few rows and columns to the constraint matrix. This may not be desirable for a sparse LP, since the new rows and columns will not generally be sparse. Furthermore, if the Newton directions are computed based on the feasibility assumption; i.e. that $r_b =0$ and $r_c = 0$, round-off errors can cause instabilities due to the gradual loss of feasibility. As an alternative, an infeasible variation of the primal-dual path-following algorithm is often used, where any $x^0>0$, $y^0$, and $z^0>0$ can be used for initialization. This algorithm will simultaneously try to reduce the duality gap and the primal-dual feasibility gap to zero. Consequently, the termination criterion will change: we stop the algorithm if $(x^k)^T z^k < \epsilon$, $||r_b||< \delta_P$, and $||r_c|| < \delta_D$.

\subsection{Computing the Newton Directions: Preconditioned Conjugate Gradient Method}
The most complex step at each iteration of the interior-point algorithm in the previous subsection is to solve the ``normal'' system of linear equations in (\ref{Delta y}). While these equations were derived for the primal-dual path-following algorithm, in most other variations of interior-point methods, we encounter linear systems of similar forms, as well. 

Various algorithms for solving linear systems fall into two main categories of \emph{direct methods} and \emph{iterative methods}. While direct methods, such as Gaussian elimination attempt to solve the system in a finite number of steps, and are exact in the absence of rounding errors, iterative methods start from an initial guess, and derive a sequence of approximate solutions. Since the constraint matrix $A D^2 A^T$ in (\ref{Delta y}) is symmetric and positive definite, the most common direct method for solving this problem is based on computing the Cholesky decomposition of this matrix. However, this approach is inefficient for large-scale sparse problems, due to the computational cost of the decomposition, as well as loss of sparsity. Hence, in many LP problems, e. g. network flow linear programs, iterative methods such as the conjugate gradient (CG) method \cite{CGMref} are preferred.

Suppose we want to find the solution $x^*$ to a system of linear equations given by
\begin{equation}
\label{linear system}
Q x = w,
\end{equation}
where $Q$ is a $q \times q$ symmetric positive definite matrix. Equivalently, $x^*$ is the unique minimizer of the functional
\begin{equation}
\label{quadratic function}
f(x) = \frac{1}{2} x^T Q x - w^T x.
\end{equation}
We call two nonzero vectors, $u, v \in \mathcal R^q$, $Q$-conjugate if
\begin{equation}
u^T Q v = 0.
\end{equation}
The CG method is based on building a set of $Q$-conjugate basis vectors $h_1, \ldots, h_q$, and computing the solution $x^*$ as
\begin{equation}
x^* = \alpha_1 h_1, \ldots, \alpha_q h_q,
\end{equation}
where $\alpha_k = \frac{h_k^T w}{h_k^T Q h_k}$.
Hence, the problem becomes that of finding a suitable set of basis vectors. In the CG method, these vectors are found in an iterative way, such that at step $k$, the next basis vector $h_k$ is chosen to be the closest vector to the negative gradient of $f(x)$ at the current point $x^k$, under the condition that it is $Q$-conjugate to $h_1, \ldots, h_{k-1}$. For a more comprehensive description of this algorithm, the reader is referred to \cite{Saad book}.

While in principle the CG algorithm requires $q$ steps to find the exact solution $x^*$, sometimes a much smaller number of iterations provides a sufficiently accurate approximation to the solution. The distribution of the eigenvalues of the coefficient matrix $Q$ has a crucial effect on the convergence behavior of the CG method (as well as many other iterative algorithms). In particular, it is shown that \cite[Chapter 6]{Saad book}
\begin{equation}
\|x^* - x^k\|_Q \leq 2 \big[ \frac{\sqrt{\kappa(Q)} - 1}{\sqrt{\kappa(Q)} + 1} \big]^k \|x^* - x^0\|_Q,
\end{equation}
where $\|x\|_Q = \sqrt{(x^T Q x)}$ and $\kappa(Q)$ is the spectral condition number of Q, i.e. the ratio of the maximum and minimum eigenvalues of Q. 
Using this result, the number of iterations of the CG method required to reduce $\|x^* - x^k\|$ by a certain factor from its initial value can be upper-bounded by a constant times $\sqrt{\kappa(Q)}$. We henceforth call a matrix $Q$ ill-conditioned, in loose terms, if CG converges slowly in solving (\ref{linear system}).

In the interior-point algorithm, the spectral behavior of $Q = A D^2 A^T$ changes as a function of the diagonal elements $d_1,\ldots, d_q,$ of $D$, which are, as described in the previous subsection, the square roots of the ratios between the primal variables $\{x_i\}$ and the dual slack variables $\{z_i\}$. In Fig. \ref{evolution of d_i}, the evolution of the distributions of $\{x_i\}$, $\{z_i\}$, and $\{d_i\}$ through the iterations of the interior-point algorithm is illustrated for an LP subproblem of an MALP decoding instance. We can observe in this figure that $x_i$ and $z_i$ are distributed in such a way that the product $x_i z_i$ is relatively constant over all $i=1,\ldots,q$. This means that, although the path-following algorithm does not completely solve the barrier problems defined in IV-B, the condition (\ref{slackness}) is approximately satisfied for all $i$. A consequence of this, which can also be observed in Fig. \ref{evolution of d_i}, is that
\begin{equation}
d_i \approx \frac{1}{\sqrt{\mu}} x_i,\ \forall\ i=1, \ldots, q.
\end{equation}
As the iterates of the interior-point algorithm become closer to the solution and $\mu$ approaches zero, many of the $d_i$'s take very small or very large values, depending on the value of the corresponding $x_i$ in the solution. This has a negative effect on the spectral behavior of $Q$, and as a result, on the convergence of the CG method. 
\begin{figure}
\centering
\includegraphics[width=5.0 in] {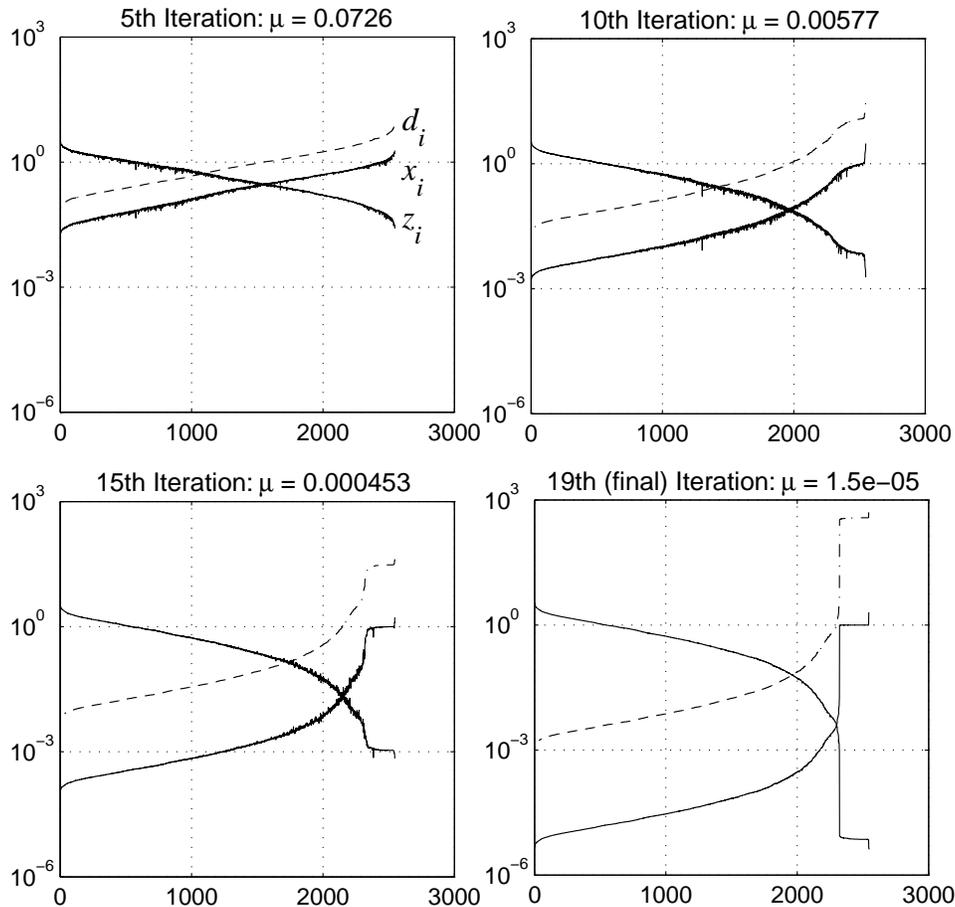}
\caption{The parameters $d_i$, $x_i$, and $z_i$, for $i=1, \ldots, q$ at four iterations of the interior-point method for an LP subproblem of MALP decoding with $n=1920$, $p=627$, $q=2547$. The variable indices, $i$, (horizontal axis) are permuted to sort $d_i$ in increasing order.}
\label{evolution of d_i}
\end{figure}

When the coefficient matrix $Q$ of the system of linear equations is ill-conditioned, it is common to use preconditioning. In this method, we use a symmetric positive-definite matrix $M$ as an approximation of $Q$, and instead of (\ref{linear system}), we solve the equivalent preconditioned system
\begin{equation}
\label{preconditioned system}
M^{-1} Q x = M^{-1} w.
\end{equation}
We hence obtain the preconditioned conjugate gradient (PCG) algorithm, summarized as Algorithm \ref{PCG Algorithm}.

\begin{algorithm}
\caption{Preconditioned Conjugate Gradient (PCG)}
\label{PCG Algorithm}
\begin{algorithmic}[1]
\STATE Compute an initial guess $x^0$ for the solution;
\STATE $r^0 = w - Q x^0$;
\STATE Solve $M z^0 = r^0$;
\STATE $h^0 = z^0$;
\FOR{$i=0, \ldots,$ until convergence}
\STATE $l^i = Q h^i$;
\STATE $\alpha_i = (z^i)^T r^i / (h^i)^T l^i$;
\STATE $x^{i+1} = x^i + \alpha^i h^i$;
\STATE $r^{i+1} = r^i - \alpha^i l^i$;
\STATE Solve $M z^{i+1} = r^{i+1}$;
\STATE $\nu^i = (z^{i+1})^T r^{i+1} / (z^i)^T r^i$;
\STATE $h^{i+1} = z^{i+1} + \nu^i h^i$;
\ENDFOR
\end{algorithmic}
\end{algorithm}
%
%

In order to obtain an efficient PCG algorithm, we need the preconditioner $M$ to satisfy two requirements. First, $M^{-1} Q$ should have a better spectral distribution than $Q$, so that the preconditioned system can be solved faster than the original system. Second, it should be inexpensive to solve $M x = z$, since we need to solve a system of this form at each step of the preconditioned algorithm. Therefore, a natural approach is to design a preconditioner which, in addition to providing a good approximation of $Q$, has an underlying structure that makes it possible to solve $M x = z$ using a direct method in linear time. 

One important application of the PCG algorithm is in interior-point implementations of LP for minimum-cost network flow (MCNF) problems. For these problems, the constraint matrix $A$ in the primal LP corresponds to the node-arc adjacency matrix of the network graph. In other words, the LP primal variables represent the edges, each constraint is defined for the edges incident to a node, and the diagonal elements, $d_1,\ldots, d_q,$ of the diagonal matrix $D$ can be interpreted as weights for the $q$ edges (variables). A common method for designing a preconditioner for $A D^2 A^T$ is to select a set $\mathcal{M}$ of $p$ columns of $A$ (edges) with large weights, and form $M= A_{\mathcal{M}} D_{\mathcal{M}}^2 A_{\mathcal{M}}^T$, where the subscript $\mathcal{M}$ for a matrix denotes a matrix consisting of the columns of the original matrix with indices in $\mathcal{M}$.

It is known that at a non-degenerate solution to an MCNF problem, the nonzero variables (i.e., the basic variables) correspond to a spanning tree in the graph. This means that, when the interior-point method approaches such a solution, the weights of all the edges, except those defining this spanning tree, will go to zero. Hence, a natural selection for $\mathcal{M}$ would be the set of indices of the spanning tree with the maximum total weight, which results in the maximum-weight spanning tree (MST) preconditioner. Finding the maximum-weight spanning tree in a graph can be done efficiently in linear time, and besides, due to the tree structure of the graph represented by $A_{\mathcal{M}}$, the matrix $M$ can be inverted in linear time as well.\footnote{Throughout the paper, we refer to solving a system of linear equations with coefficient matrix $M$, in loose terms, as inverting $M$, although we do not explicitly compute $M^{-1}$.} The MST has been observed in practice to be very effective, especially at the latter iterations of the interior-point method, when the operating point is close to the final solution.

\section{Preconditioner Design for LP Decoding}
Our framework for designing an effective preconditioner for LP decoding, similar to the MST preconditioner for MCNF problems, is to find a \emph{preconditioning set}, $\mathcal{M} \subseteq \{1,\ldots, q\}$, corresponding to $p$ columns of $A$ and $D$, resulting in $p \times p$ matrices $A_{\mathcal{M}}$ and $ D_{\mathcal{M}}$, such that $M= A_{\mathcal{M}} D_{\mathcal{M}}^2 A_{\mathcal{M}}^T$ is both easily invertible and a good approximation of $Q = A D^2 A^T$. To satisfy these requirements, it is natural to select $\mathcal{M}$ to include the variables with the highest weights, $\{d_i\}$, while keeping $A_\mathcal{M}$ and $D_\mathcal{M}$ full rank and invertible in $O(q)$ time. Then, the solution $z^{i+1}$ to $M z^{i+1} = r^{i+1}$ in the PCG algorithm can be found by sequentially solving $A_{\mathcal{M}} f_1 = r^{i+1}$, $D_{\mathcal{M}}^2 f_2 = f_1$, and $A_{\mathcal{M}}^T z^{i+1} = f_2$, for $f_1$, $f_2$, and $z^{i+1}$, respectively.

We are interested in having a graph representation for the constraints and variables of a linear program of the form (\ref{primal LP}) in the LP decoding problem, such that the selection of a desirable $\mathcal{M}$ can be interpreted as searching for a subgraph with certain combinatorial structures. 

\begin{definition}
\label{extended Tanner}
Consider an LP of the form (\ref{primal LP}) with $p$ constraints and $q$ variables, where $x_{n+1}, \ldots, x_q$ are slack variables. The \emph{extended Tanner graph} of this LP is a bipartite graph consisting of $q$ \emph{variable nodes} and $p$ \emph{constraint nodes}, such that variable node $i$ is connected to constraint node $i$ if $x_i$ is involved in the $j$th constraint; i.e., $A_{i,j}$ is nonzero.
\end{definition}

For the linear programs in the MALP decoding algorithms, since each constraint is derived from a unique check node of the original Tanner graph, the extended Tanner graph will be a subgraph of the Tanner graph, with the addition of $q$ degree-1 (slack) variable nodes, each connected to one of the constraint nodes. In general, for an iteration of MALP decoding of a code with an $m\times n$ parity-check matrix, the extended Tanner graphs would contain $p \leq m$ constraint nodes, $n$ variable nodes corresponding to the standard variables (bit positions), and $p$ slack variable nodes. As extended Tanner graphs are special cases of Tanner graphs, they inherit all the combinatorial concepts defined for Tanner graphs, such as stopping sets. A small example of an extended Tanner graph is given in Fig. \ref{Extended Tanner fig}.

\begin{figure}
\centering
\includegraphics[width=2.2 in] {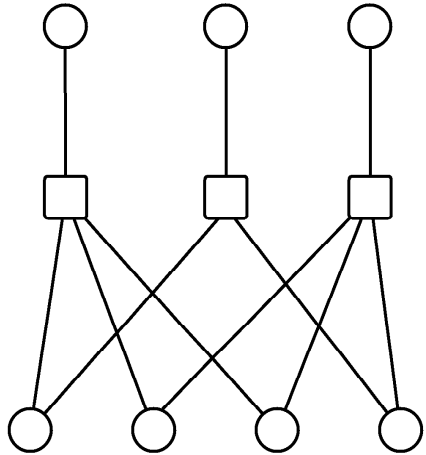}
\caption{An extended Tanner graph for an LP problem with $n=4$, $p=3$, and $q=7$.}
\label{Extended Tanner fig}
\end{figure}

\subsection{Preconditioning via Triangulation}
For a sparse constraint matrix, $A$, a sufficient condition for $A_\mathcal{M}$ and $A_\mathcal{M}^T$ to be invertible in $O(q)$ time is that $A_\mathcal{M}$ can be made upper or lower triangular, with nonzero diagonal elements, using column and/or row permutations. We call a preconditioning set $\mathcal{M}$ that  satisfies this property a \emph{triangular set}. Once an upper- (lower-) triangular form $A_\mathcal{M}^\triangle{}$ of $A_\mathcal{M}$ is found, we start from the last (first) row of $A_\mathcal{M}^\triangle{}$, and, by taking advantage of the sparsity, solve for the variable corresponding to the diagonal element of each row recursively in $O(1)$ time. It is not difficult to see that there always exists at least one triangular set for any LP decoding problem; one example is the set of columns corresponding to the slack variables, which results in a diagonal $A_\mathcal{M}$.

As a criterion for finding the best approximation $A_\mathcal{M} D_\mathcal{M}^2 A_\mathcal{M}^T$ of $A D^2 A^T$, we search for the triangular set that contains the columns with the highest weights, $d_i$. One can consider different strategies of scoring a triangular set from the weights of its members, e.g., the sum of the weights, or the largest value of minimum weight. It is interesting to study as a future work whether given any such metric, the ``maximum-weight'' (or optimal) triangular set can be found in polynomial time. However, in this work, we propose a (suboptimal) greedy approach, which is motivated by the properties of the LP decoding problem.

The problem of bringing a parity-check matrix into (approximate) triangular form has been studied by Richardson and Urbanke \cite{Efficient encoding} in the context of the encoding of LDPC codes. The authors proposed a series of greedy algorithms that are similar to the peeling algorithm for decoding in the binary erasure channel: repeatedly select a nonzero entry (edge) of the matrix (graph) lying on a degree-1 column or row (variable or check node), and remove both the column and row of this entry from the matrix. They showed that parity-check matrices that are optimized for erasure decoding can be made almost triangular using this greedy approach. 
It is important to note that this combinatorial approach only relies on the placement of the nonzero entries of the matrix, rather than their values. 

The fact that the constraint matrices of the LP problems in MALP decoding have structure similar to the corresponding parity-check matrix motivates the use of a greedy algorithm analogous to those in \cite{Efficient encoding} for triangulating the matrix $A$. 
However, this problem is different from the encoding problem, in that we are not merely interested in making $A$ triangular, but rather, we look for the triangular submatrix with the maximum weight. In fact, as mentioned earlier, finding one triangular form of $A$ is trivial, due to the presence of the slack variables. Here, we present three greedy algorithms to search for the MTS, one of which is related to the algorithms of Richardson and Urbanke. Throughout this section, we will also refer to the outputs of these (suboptimal) greedy algorithms, in loose terms, as the MTS, although they may not necessarily have the maximum possible weight.

\subsubsection{Incremental Greedy Search for the MTS}
Although an ideal preconditioning set would contain the $q$ columns of the matrix that have the $q$ highest weights, in reality, the square submatrix of $A$ composed of these $q$ columns is often neither triangular nor full rank. In the incremental greedy search for the MTS, we start by selecting the highest-weight column, and try to expand the set of selected columns by giving priority to the columns of higher weights, while maintaining the property that the corresponding submatrix can be made lower-triangular by column and row permutations.

Let $\mathcal{S}$ be a set of selected columns from $A$, where $|\mathcal{S}| \leq p$. In order to check whether the submatrix $A_\mathcal{S}$ can be made lower-triangular by column and row permutations, we can treat the variable nodes corresponding to $\mathcal{S}$ in the Tanner graph as erased bits, and use the peeling algorithm to decode them in $O(q)$ time. For completeness, this process, which we call the Triangulation Step, is described in Algorithm~\ref{triangulate}.

\begin{algorithm}
\caption{Triangulation Step}
\label{triangulate}
\begin{algorithmic}[1]
\STATE {\bf Input:} The set $\mathcal{S}$ with $|\mathcal{S}| = s \leq p$, and the matrix $A$;
\STATE {\bf Output:} An $s \times s$ lower-triangular submatrix $A_\mathcal{S}^\triangle{}$, if possible;
\STATE {\bf Initialization:} $\tilde{A} \leftarrow A_\mathcal{S}$, and initialize $col$ and $row$ as zero-length vectors;
\FOR{$k=1$ to $s$}
\STATE {\bf if} the minimum row degree in $\tilde{A}$ is not one {\bf then} $A_\mathcal{S}$ cannot be made lower-triangular by permutation; Declare {\bf Failure} and exit the algorithm;
\STATE Select any degree-1 row $j$ from $\tilde{A}$, and let $i$ be the index of the column that contains the only nonzero entry of row $j$;
\STATE $col \leftarrow {\scriptsize \begin{bmatrix} col \\ i \end{bmatrix}}$ , $row \leftarrow {\scriptsize \begin{bmatrix} row \\ j \end{bmatrix}}$;
\STATE Set all the entries in column $i$ and row $j$ of $\tilde{A}$ to zero;
\ENDFOR
\STATE Form $A_\mathcal{S}^\triangle{}$ by setting ${A_\mathcal{S}^\triangle{}}_{i,j} = {A_\mathcal{S}}_{col_i , row_j}$, $\forall\ i,j \in \{1,\ldots s\}$;
\end{algorithmic}
\end{algorithm}

Using the Triangulation Step as a subroutine, the incremental greedy search method, given by Algorithm~\ref{MTS inc}, first sorts the columns according to their corresponding weights, $d_i$ (or, alternatively, $x_i$), and initializes the preconditioning set, $\mathcal{M}$, as an empty set. Starting with the highest-weight column and going down the sorted list of column indices, it adds each column to $\mathcal{M}$ if the submatrix corresponding to the resulting
set can be made lower triangular using the Triangulation Step.

\begin{algorithm}
\caption{Incremental Greedy Search for the MTS}
\label{MTS inc}
\begin{algorithmic}[1]
\STATE {\bf Input:} $p \times q$ constraint matrix $A$, and the set of column weights, $d_1 \ldots d_q$;
\STATE {\bf Output:} A triangular set $\mathcal{M}$ and the $p \times p$ lower-triangular matrix $A_\mathcal{M}^\triangle{}$;
\STATE {\bf Initialization:} $\mathcal{M} \leftarrow \emptyset$, $i \leftarrow 0$;
\STATE Sort the column indices $\{1, \ldots, q\}$ according to their corresponding weights, $d_i$, in decreasing order, to obtain the permuted sequence $\pi_1, \ldots, \pi_q$, such that $d_{\pi_1} \geq \ldots \geq d_{\pi_q}$;
\WHILE{$i < q$ and $|\mathcal{M}|<p$}
\STATE $i \leftarrow i+1$, $\mathcal M \leftarrow \mathcal{M} \cup \{\pi_i\}$;
\IF {the Triangulation Step can bring the submatrix $A_\mathcal{S}$ into the lower-triangular form $A_\mathcal{S}^\triangle{}$} 
\STATE $\mathcal{M}\leftarrow \mathcal{S}$, $A_\mathcal{M}^\triangle{} \leftarrow A_\mathcal{M}^\triangle{}$;
\ENDIF
\ENDWHILE
\end{algorithmic}
\end{algorithm}

We claim that, due to the presence of the slack columns in $A$, Algorithm \ref{MTS inc} will successfully find a triangular set $\mathcal{M}$ of $p$ columns; i.e., it exits the while-loop (lines 5-10) only when $|\mathcal{M}| = p$. Assume, on the contrary, that the algorithm ends while $|\mathcal{M}| < p$, so that the matrix $A_\mathcal{M}$ is a $p \times |\mathcal{M}|$ lower-triangular matrix. This means that if we add any column $k \in \{1,\ldots, q\} \backslash \mathcal{M}$ to $\mathcal{M}$, it cannot be made lower triangular, since otherwise, column $k$ would have already been added to $|\mathcal{M}|$ when $\pi_i = k$ in the while-loop.\footnote{Note that if any set $\mathcal{S}$ of columns can be made lower triangular, any subset of these columns can be made lower triangular, as well.} 
However, this clearly cannot be the case, since we can produce a $p \times p$ lower-triangular matrix $A_\mathcal{M}^\triangle{}$, simply by adding the columns corresponding to the slack variables of the last $p - |\mathcal{M}|$ rows of $A_\mathcal{M}$. Hence, we conclude that $|\mathcal{M}| = p$.

\subsubsection{Column-wise Greedy Search for the MTS}
Algorithm \ref{MTS col} is a column-wise greedy search for the MTS. It successively adds the index of the maximum-weight degree-1 column of $A$ to the set $\mathcal{M}$, and eliminates this column and the row that shares its only nonzero entry. Matrix $A$ initially contains $p$ degree-1 slack columns, and at each iteration, one such column will be erased. Hence, there is always a degree-1 column in the residual matrix, and the algorithm proceeds until $p$ columns are selected. The resulting preconditioning set will correspond to an upper-triangular submatrix $A_\mathcal{M}$.

\begin{algorithm}
\caption{Column-wise Greedy Search for the MTS}
\label{MTS col}
\begin{algorithmic}[1]
\STATE {\bf Input:} $p\times q$ constraint matrix $A$, and the set of column weights $d_1, \ldots, d_q$;
\STATE {\bf Output:} A triangular set $\mathcal{M}$ and the upper-triangular matrix $A_\mathcal{M}^\triangle{}$;
\STATE {\bf Initialization:} $\tilde{A} \leftarrow A$, $\mathcal{M} \leftarrow \emptyset$, and initialize $col$ and $row$ as zero-length vectors;
\STATE Define and form $\mathcal{DEG}1$ as the index set of all degree-1 columns in $\tilde{A}$;
\FOR{$k=1$ to $p$}
\STATE Let $i \in \mathcal{DEG}1$ be the index of the (degree-1) column of $\tilde{A}$ with the maximum weight, $d_i$, and let $j$ be the index of the row that contains the only nonzero entry of this column;
\STATE $\mathcal{M} \leftarrow \mathcal{M}\cup i$,  $col \leftarrow {\scriptsize \begin{bmatrix} col \\ i \end{bmatrix}}$ , $row \leftarrow {\scriptsize \begin{bmatrix} row \\ j \end{bmatrix}}$;
\STATE Set all the entries in row $j$ of $\tilde{A}$ (including the only nonzero entry of column $i$) to zero;
\STATE Update $\mathcal{DEG}1$ from the residual matrix, $\tilde{A}$;
\ENDFOR
\STATE Form $A_\mathcal{M}^\triangle{}$ by setting ${A_\mathcal{M}^\triangle{}}_{i,j} = A_{col_i , row_j}$, $\forall\ i,j \in \{1,\ldots p\}$;
\end{algorithmic}
\end{algorithm}

\subsubsection{Row-wise Greedy Search for the MTS}
Algorithm \ref{MTS row} uses a row-wise approach for finding the MTS. In this method, we look at the set of degree-1 rows, add to $\mathcal{M}$ the indices of all the columns that intersect with these rows at nonzero entries, and eliminate these rows and columns from $A$. Unlike the column-wise method, it is possible that, at some iteration, these is no degree-1 row in the matrix. In this case, we repeatedly eliminate the lowest-weight column, until there is one or more degree-1 rows. 

\begin{algorithm}
\caption{Row-wise Greedy Search for the MTS}
\label{MTS row}
\begin{algorithmic}[1]
\STATE {\bf Input:} $p\times q$ constraint matrix $A$, and the set of column weights $d_1, \ldots, d_q$;
\STATE {\bf Output:} A triangular set $\mathcal{M}$ and the lower-triangular matrix $A_\mathcal{M}^\triangle{}$;
\STATE {\bf Initialization:} $\tilde{A} \leftarrow A$, $\mathcal{M} \leftarrow \emptyset$, and initialize $col$ and $row$ as zero-length vectors;
\STATE Define and form $\mathcal{DEG}1$ as the index set of all degree-1 rows in $\tilde{A}$;
\WHILE{$\tilde{A}$ is not all zeroes}
\IF{$|\mathcal{DEG}1|>0$}
\STATE Let $j \in \mathcal{DEG}1$ be any degree-1 row of $\tilde{A}$, and $i$ be the index of the column that contains the only nonzero entry of this row;
\STATE $\mathcal{M} \leftarrow \mathcal{M}\cup i$,  $col \leftarrow {\scriptsize \begin{bmatrix} col \\ i \end{bmatrix}}$ , $row \leftarrow {\scriptsize \begin{bmatrix} row \\ j \end{bmatrix}}$;
\STATE Set all the entries in column $i$ of $\tilde{A}$ (including the only nonzero entry of row $j$) to zero, and update $\mathcal{DEG}1$;
\ELSE
\STATE Let $i$ be the index of the nonzero column of $\tilde{A}$ with the minimum weight, $d_i$. Set all the entries in column $i$ to zero, and update $\mathcal{DEG}1$;
\ENDIF
\ENDWHILE
\STATE {\bf Diagonal Expansion:} For each row $j$ of $A$ that is not represented in $row$, append $j$ to $row$, and append $i=j+n$, i.e., the index of the corresponding slack column, to both $col$ and $\mathcal{M}$;
\STATE Form $A_\mathcal{M}^\triangle{}$ by setting ${A_\mathcal{M}^\triangle{}}_{i,j} = A_{col_i , row_j}$, $\forall\ i,j \in \{1,\ldots p\}$;
\end{algorithmic}
\end{algorithm}

In addition to this difference, the number of columns in $\mathcal{M}$ by the end of this procedure is often slightly smaller that $p$. Hence, we perform a ``diagonal expansion'' step at the end, where $p - |\mathcal{M}|$ columns corresponding to the slack variables are added to $\mathcal{M}$, while keeping it a triangular set. A problem with this expansion method is that, since the algorithm does not have a choice in selecting the slack variables added in this step, it may add columns that have very small weights. 

Let $A_{\mathcal{M}_1}^\triangle{}$ be the triangular submatrix obtained before the expansion step. 
As an alternative to diagonally expanding $A_{\mathcal{M}_1}^\triangle{}$ by adding slack columns, we can apply a ``triangular expansion.'' In this method, we form a matrix $\bar{A}$ consisting of the columns of $A$ that do not share any nonzero entries with the rows in vector $row$, and apply a column-wise or row-wise greedy search to this matrix in order to obtain a high-weight lower-triangular submatrix $A_{\mathcal{M}_2}^\triangle{}$. This requirement for forming $\bar{A}$ ensures that the resulting triangular submatrices $A_{\mathcal{M}_1}^\triangle{}$ and $A_{\mathcal{M}_2}^\triangle{}$ can be concatenated as
\begin{equation}
\begin{bmatrix}
A_{\mathcal{M}_1}^\triangle{} & 0\\
B & A_{\mathcal{M}_2}^\triangle{}\\
\end{bmatrix},
\end{equation}
to form a larger triangular submatrix of $A$. This process can be continued, if necessary, until a square $p \times p$ triangular matrix $A_{\mathcal{M}}^\triangle{} $ is obtained, although our experiments indicate that one expansion step is often sufficient to provide such a result. It is easy to see that this approach is potentially stronger than the diagonal expansion in Algorithm \ref{MTS row}, since it has the diagonal expansion as a special case.

\subsection{Implementation and Complexity Considerations}
To compute the running time of Algorithm \ref{MTS inc}, note that while Step 4 has $O(q \log q)$ complexity, the computational complexity of the algorithm is dominated by the Triangulation Step. This subroutine has $O(q)$ complexity, and is called $O(q)$ times in Algorithm \ref{MTS inc}, which makes the overall complexity $O(q^2)$. An interesting problem to investigate is whether we can simplify the triangulation process in line 7 to have sublinear complexity by exploiting the results of the previous round of triangulation, as stated in the following open problem concerning erasure decoding:

\emph{Open Problem:} Consider the Tanner graph corresponding to an arbitrary LDPC code of length $n$. Assume that a set $\mathcal{E}$ of bits are erased, and $\mathcal{E}$ does not contain a stopping set in the Tanner graph. Thus, the decoder successfully recovers these erased bits using the peeling algorithm (i.e., the triangulation Algorithm \ref{triangulate}). Now, we add a bit $i$ to the set of erased bits. Given $j$, $\mathcal{E}$, and the complete knowledge of
the process of decoding $\mathcal{E}$, such as the order in which the bits are decoded, and the check nodes used, is there an $o(n)$ scheme to verify if $\mathcal{E} \cup \{i\}$ can be decoded by the peeling algorithm?

In addition this potential simplification, it is possible to make a number of modifications to Algorithm \ref{MTS inc} in order to reduce its complexity. Let $s$ be the size of the smallest stopping set in the extended Tanner graph of $A$, which means that the submatrix formed by any $s-1$ columns can be made triangular. Then, instead of initializing $\mathcal{M}$ to be the empty set, we can immediately add the $s-1$ highest-weight columns to $\mathcal{M}$, since we are guaranteed that $A_\mathcal{M}$ can be made triangular. Moreover, at each iteration of the algorithm, we can consider $k>1$ column to be added to $\mathcal{M}$, in order to reduce the number of calls to the triangulation subroutine. The value of $k$ can be adaptively selected to make sure that the modified algorithm remains equivalent to Algorithm \ref{MTS inc}.

To assess the complexity of Algorithm \ref{MTS col}, we need to examine Steps 8 and 11 that involve column or row operations, as well as Steps 4, 6, and 9 that deal with the list of degree-1 columns. 
Since there is an $O(1)$ number of nonzero entries in each column or row of $A$, running Step 8 $p$ times (due to the for-loop), and deriving $A_\mathcal{M}^\triangle{}$ from $A$ in Step 11 each take $O(q)$ time. However, one should be careful in selecting a suitable data structure for storing the set $\mathcal{DEG}1$, since, in each cycle of the for-loop, we need to extract the element with the maximum weight, and add to and remove from this set an $O(1)$ number of elements. By using a binary heap data structure \cite{Knuth}, which is implementable as an array, all these (Steps 6 and 9) can be done in $O(\log q)$ time in the worst case. Also, the initial formation of the heap (Step 4) has $O(q)$ complexity. As a result, the total complexity of Algorithm \ref{MTS col} becomes $O(q \log q)$.

Similarly, in Algorithm \ref{MTS row}, we need a mechanism to extract the minimum-weight member of the set of remaining columns. While the heap structure mentioned above works well here, since no column is added to the set of remaining columns, we can alternatively sort the set of all columns by their weights as a preprocessing step with $O(q \log q)$ complexity, thus making the complexity of the while-loop linear. Since the complexity of steps 15 (diagonal expansion) and 16 are linear, as well, the total running time of Algorithm \ref{MTS row} will be $O(q \log q)$.

The process of finding a triangular preconditioner is performed at each iteration of the interior-point algorithm. Since the values of primal variables, $\{x_i\}$, do not substantially change in one iteration, we expect the maximum-weight triangular set at each iteration to be relatively close to that in the previous iteration. Consequently, an interesting area for future work is to investigate modifications of the proposed algorithms, where the knowledge of the MTS in the previous iteration of the interior-point method is exploited to improve the complexity of these algorithms.

\section{Analysis of the MTS Preconditioning Algorithms}
\subsection{Performance Analysis}
It is of great interest to study how the proposed algorithms perform as the problem size goes to infinity. We expect that a number of asymptotic results similar to those of Richardson and Urbanke in \cite{Efficient encoding} can be derived, e.g., showing that the greedy preconditioner designs perform well for capacity-approaching LDPC ensembles. However, since one of the main advantages of LP decoding over message-passing decoding is its geometrical structure that facilitates the analysis of its performance in the finite-length regime, in this work we focus on studying the proposed algorithms in this regime.

We will study the behavior of the proposed preconditioner in the later iterations of the interior-point algorithm, when the iterates are close to the optimum. This is justified by the fact that, as the interior-point algorithm approaches the boundary of the feasible region during its later iterations, many of the primal variables, $x_i$, and the dual slack variables, $z_i$, approach zero, thus deteriorating the conditioning of the matrix $Q = A D^2 A^T$. This is when a precoditioner is most needed.
In addition, we can obtain some information about the performance of the preconditioner in the later iterations by focusing on the optimal point of the feasible set.

Consider an LP problem in the augmented form (\ref{primal LP}) as part of ALP or MALP decoding, and assume that it has a unique optimal solution (although parts of our analysis can be extended to the case with non-unique solutions). We denote by the triplet $(x^*, y^*, z^*)$ the primal-dual solution to this LP, and by $(x, y, z)$ an intermediate iterate of the interior-point method. We can partition the set of the $q$ columns of $A$ into the \emph{basic set} 
\begin{equation}
\mathcal{B}=\{i | x_i^* >0\}
\end{equation}
and the \emph{nonbasic set}
\begin{equation}
\mathcal{N}=\{i | x_i^* =0\}.
\end{equation}

For brevity, we henceforth refer to the columns of the constraint matrix $A$ corresponding to the basic variables as the ``basic columns.'' It is not difficult to show that, for an LP with a unique solution, the number of basic variables, i.e., $|\mathcal{B}|$, is at most $p$. To see this, assume that $l$ of the standard variables $x_1^* \ldots x_n^*$ are nonzero, which means that $n-l$ box constraints of the form $x_i \geq 0$ are active at $x^*$. Since $x^*$ is a vertex defined by at least $n$ active constraints in the LP, we conclude that at least $l$ parity inequalities must be active at $x^*$, thus leaving at most $p-l$ nonzero slack variables. We call the LP \emph{nondegenerate} if $|\mathcal{B}| = p$, and \emph{degenerate} if $|\mathcal{B}| < p$.

It is known that the unique solution $(x^*, y^*, z^*)$ is ``strictly complementary'' \cite{Goldman}, meaning that for any $i \in \{1, \ldots, q\}$ either $x_i^*=0$ and $z_i^* >0$, or $x_i^*>0$ and $z_i^*=0$. Remembering from (\ref{definition of D2}) that $d_i = \sqrt{x_i / z_i}$, as the iterates of the interior-point algorithm approach the optimum, i.e., $\mu$ given in (\ref{update mu}) goes to zero, we will have
\begin{equation}
\label{limit of d_i}
\setlength{\nulldelimiterspace}{0pt}
\lim_{\mu \rightarrow 0} d_i =
\left\{\begin{IEEEeqnarraybox} [\relax][c]{l's}
\infty \hspace{0.2in} \text{if} \ \ i \in \mathcal{B}, \\
\ 0\ \hspace{0.3in} \text{if} \ \ i \in \mathcal{N},
\end{IEEEeqnarraybox}\right.
\end{equation}
Therefore, towards the end of the algorithm, the matrix $Q=A D^2 A^T$ will be dominated by the columns of $A$ and $D$ corresponding to the basic set. Hence, it is highly desirable to select a preconditioning set that includes all the basic columns, i.e., $\mathcal{B} \subseteq \mathcal{M}$, in which case $A_\mathcal{M} D_\mathcal{M}^2 A_\mathcal{M}^T$ becomes a better and better approximation of $Q$, as we approach the optimum of the LP. 
In the rest of this subsection, we will show that, when the solution to the LP is integral and $\mu$ is sufficiently small, this property can be achieved by low complexity algorithms similar to Algorithms \ref{MTS col} and \ref{MTS row}.

\begin{lemma}
\label{Basic set is triangular}
Consider the extended Tanner graph $\mathcal{T}^k$ for an LP subproblem $LP^k$ of MALP decoding. If the primal solution to $LP^k$ is integral, the set of variable nodes corresponding to the basic set, whose definition is based on the augmented form (\ref{primal LP}) of the LP, does not contain any stopping set.
\end{lemma}
\begin{proof}
Consider an erasure decoding problem $P^{BEC}$ on $\mathcal{T}^k$, where the basic variable nodes are erasures. We prove the lemma by showing that the peeling (or LP) decoder can successfully correct these erasures.

We denote by $u^*$ and $x^*$ the solutions to the primal LP in the (original) standard form (\ref{MALP matrix form}) and in the augmented form (\ref{primal LP}). From part c) of Theorem \ref{properties of ALP}, we know that $u^*$ is also the solution to a full LP decoding problem $LPD^k$ with the LLR vector $\gamma$ and the Tanner graph comprising the standard variable nodes and the active check nodes, $\mathcal{J}_{act}$. 

We partition the basic set $\mathcal{B}$ into $\mathcal{B}_{std}$ and $\mathcal{B}_{slk}$, the sets of basic standard variables and basic slack variables, respectively. We also partition the set of check nodes in $\mathcal{T}^k$ into $\mathcal{J}_{act}$ and $\mathcal{J}_{inact}$, the sets of check nodes that generate the active and inactive parity inequalities of $LP^k$, respectively. Clearly, the neighbors of the slack variable nodes in $\mathcal{B}_{slk}$ are the check nodes in $\mathcal{J}_{inact}$, since an inactive parity inequality has, by definition, a nonzero slack.

\emph{Step 1:} We first show that, even if we remove the check nodes in $\mathcal{J}_{inact}$ from $\mathcal{T}^k$, the set of basic standard variable nodes, $\mathcal{B}_{std}$, does not contain a stopping set. 

Remembering the conversion of the LP in the standard form (\ref{MALP matrix form}) with inequality constraints to the augmented form (\ref{primal LP}), we can write
\begin{align}
\label{B_std}
\mathcal{B}_{std} = \big\{i \in \mathcal{I}\ \big|\ (\gamma_i \geq 0\ ,\ u^*_i=1) \text{ or } (\gamma_i < 0\ ,\ u^*_i=0) \big\}.
\end{align}
Using, as in Theorem \ref{connection to erasure}, the notation $\tilde{u}$ for the result of bit-based hard decision on $\gamma$, one can see that $\mathcal{B}_{std}$ is identical to $\mathcal{E}$, the set of positions where $u^*$ and $\tilde{u}$ differ. Hence, knowing that $u^*$ is the solution to an LP decoding problem, and using Theorem \ref{connection to erasure}, we conclude that the set $\mathcal{B}_{std}$ does not contain a stopping set in the Tanner graph that only includes the check nodes in $\mathcal{J}_{act}$.

\emph{Step 2:} Now we return to $\mathcal{T}^k$, and consider solving $P^{BEC}$, where all the basic variables are erasures, using the peeling algorithm. Since the slack variables which are basic are connected only to the inactive check nodes, we know from Step 1 that the erased variables $\mathcal{B}_{std}$ can be decoded by only using the active check nodes $\mathcal{J}_{act}$. Once these variable nodes are peeled off the graph, we are left with the basic slack variable nodes, each of which is connected to a distinct check node in $\mathcal{J}_{inact}$. Therefore, the peeling algorithm can proceed by decoding all of these variables. This completes the proof.
\end{proof}

Lemma \ref{Basic set is triangular} shows that, under proper conditions, the submatrix $\tilde{A}$ of $A$ formed by only including the columns corresponding to the basic variables can be made lower triangular by column and row permutations. This suggests that looking for a maximum-weight triangular set is a natural approach for designing a preconditioner in MALP decoding. In particular, the following theorem shows that, under the conditions of Lemma \ref{Basic set is triangular}, the incremental greedy Algorithm \ref{MTS inc} indeed finds a preconditioning set that includes all such columns.

As the interior-point algorithm progresses, the basic variables approach 1, while the nonbasic variables approach zero. Hence, referring to (\ref{limit of d_i}), we see that after a large enough number of iterations, the $|\mathcal{B}|$ highest-weight columns of $A$ will correspond to the basic set $\mathcal{B}$. The following theorem shows that two of the proposed algorithms indeed find a preconditioning set that includes all such columns.

\begin{theorem}
\label{Optimality of MTS}
Consider an LP subproblem $LP^k$ of an MALP decoding problem. If the primal solution to $LP^k$ is integral, at the iterates of the interior-point method that are sufficiently close to the solution, both the Incremental Greedy Algorithm and the Row-wise Greedy Algorithm can successfully find a triangular set that includes all the columns corresponding to the basic set.
\end{theorem}
\begin{proof}
As the interior-point algorithm progresses, the weights $d_i$ corresponding to the basic variables approach $\infty$, while the weights of nonbasic variables approach zero.
Hence, when $\mu$ becomes sufficiently small, the columns corresponding to the basic set, $\mathcal{B}$ will be the $|\mathcal{B}|$ highest-weight columns of $A$, and according to Lemma \ref{Basic set is triangular}, the matrix $A_\mathcal{B}$ consisting of these columns can be made triangular, provided that the solution to $LP^k$ is integral.

In view of this result, the proof of the claim for the incremental greedy algorithm becomes straighforward: The preconditioning set $\mathcal{M}$ continues to grow by one member at each iteration, at least until it includes all the $|\mathcal{B}|$ highest-weight (i.e., basic) columns.

To prove that the triangular set $\mathcal{M}$ given by the row-wise greedy algorithm includes the basic set, as well, it is sufficient to show that none of the basic columns will be erased from $\tilde{A}$ (i.e., become all zeroes) in line 11 of Algorithm \ref{MTS row}. 
Assume that, at some iteration, a column $i$ is selected in line 11 to be erased. 
Column $i$ has the minimum weight among the nonzero columns currently in $\tilde{A}$. Therefore, if $i$ is a basic column and $\mu$ is small enough, all the other nonzero columns are basic columns, as well, since the basic columns are the $|\mathcal{B}|$ highest-weight columns of $A$. This means that $\tilde{A}$ could be made triangular, without running out of degree-1 rows and having to erase column $i$. So, column $i$ cannot be basic. 
\end{proof}

\begin{remark}
The proof above suggests that Theorem \ref{Optimality of MTS} can be stated in more general terms.
For any $s \in \{1, \ldots, q\}$, let $\mathcal{S}$ be a set consisting of the $s$ highest-weight columns of $A$. Then, if the set of variable nodes corresponding to $\mathcal{S}$ in the (extended) Tanner graph does not contain a stopping set, that is, $A_\mathcal{S}$ can be made triangular by row and column permutations, then the preconditioning sets found by Algorithms~\ref{MTS inc} and~\ref{MTS row} both contain $S$.
\end{remark}

The assumption that the solution is integral does not hold for all LPs that we solve in adaptive LP decoding. On the other hand, in practice, we are often interested in solving the LP exactly, only when LP decoding finds an integral solution (i.e., the ML codeword). This, of course, does not mean that in such cases every LP subproblem solved in the adaptive method has an integral solution. However, one can argue heuristically that, if the final LP subproblem has an integral solution, the intermediate LPs are also very likely to have an integral solution. To see this, remember from Theorem \ref{properties of ALP} that each intermediate LP problem that is solved in adaptive LP decoding is equivalent to a full LP decoding that uses a subset of the check nodes in the Tanner graph. Now, if LP decoding with the complete Tanner graph has an integral solution, it is natural to expect that, after removing a subset of check nodes, which can also reduce the number of cycles, the LP decoder still very likely to find an integral solution.

\subsection{Performance Simulation}
We simulated the LP decoding of $(3,6)$-regular LDPC codes on the AWGN channel using
the MALP-A algorithm and our sparse implementation of the path-following interior-point method. We have shown earlier that, as interior-point progresses, the matrix $A D^2 A^T$ that needs to be inverted to compute the Newton steps becomes more and more ill-conditioned. We have observed that this problem becomes more severe in the later iterations of the MALP-A algorithm, where the LP problem is larger and more degenerate due to the abundance of active constraints at the optimum of the problem. 

In Figs. \ref{PCGM_int1}-\ref{PCGM_frac2}, we present the performance results of the PCG method for four different systems of linear equations in the form of (\ref{Delta y}), solved in the infeasible primal-dual path-following interior-point algorithm, using the preconditioners designed by greedy Algorithms \ref{MTS inc}-\ref{MTS row}.\footnote{In all the simulations of the Row-wise Greedy Search (Algorithm \ref{MTS row}) that we present in this section, we have used a diagonal expansion, rather than a triangular expansion, as described in Subsection V-A.} In these simulations, we used a randomly-generated $(3,6)$-regular LDPC code of length 2000, where the cycles of length four were removed. The performance of the PCG algorithm is measured by the behavior of the relative residual error ${\|r^i\|_2^2}/{\|w\|_2^2}$, where $r^i$ and $w$ are defined in Algorithm \ref{PCG Algorithm}, as a function of the iteration number $i$ of the PCG algorithm.

In Figs. \ref{PCGM_int1} and \ref{PCGM_int2}, we considered solving (\ref{Delta y}) in two different iterations of the interior-point algorithm for solving an LP problem. This LP problem was selected at the 6th iteration of an MALP decoding problem at SNR = 1.5 dB, and the solution to the LP was \emph{integral}. The constraint matrix $A$ for this LP had 713 rows and 2713 columns, and we used the PCG algorithm to compute the Newton step. Fig. \ref{PCGM_int1} corresponds to finding the Newton step at the $8$th iteration of the interior-point algorithm. In this scenario, the duality gap $g_d = x^T z$ was equal to 48.6, and the condition number $\kappa(Q)$ of the problem was equal to $3.46 \times 10^4$. We have plotted the residual error of the CG method without preconditioning, as well as the PCG method using the three proposed preconditioner designs. For this problem, except during the first 10-15 iterations, the behaviors of the three preconditioned implementations are very similar, and all significantly outperform the CG method.

In Fig. \ref{PCGM_int2}, we solved (\ref{Delta y}) at the $18$th iteration of the same LP, where the interior-point is much closer to the solution, with $g_d = 0.22$ and $\kappa(Q) = 2.33 \times 10^8$. In this problem, the convergence of the CG method is very slow, so that in 200 iterations, the residual error does not get below $0.07$. The PCG method with incremental greedy preconditioning, reaching a residual error of $10^{-4}$ in 40 iterations, has the fastest convergence, followed by the column-wise greedy preconditioner. 

\begin{figure}
\centering
\includegraphics[width=4.5 in] {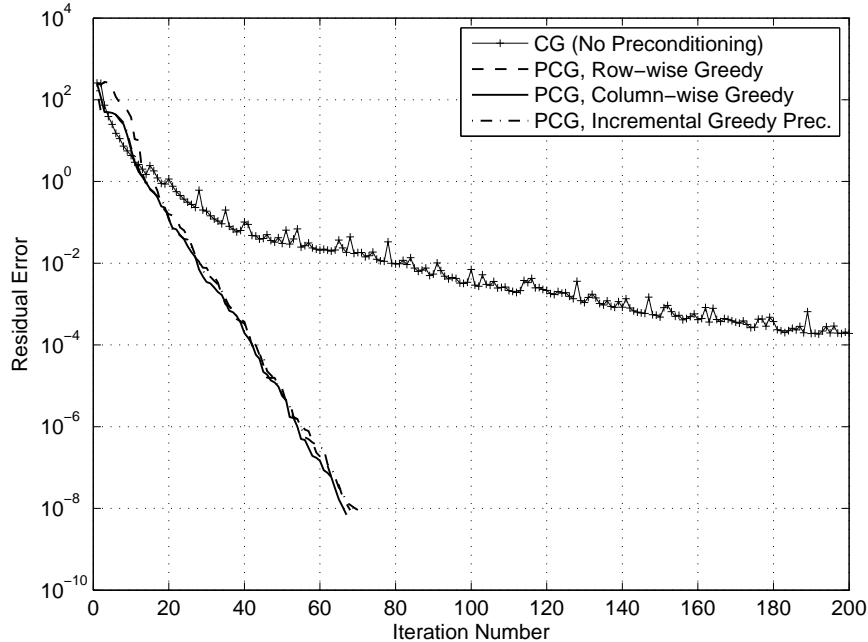}
\caption{The progress of the residual error for different PCG implementations, solving (\ref{Delta y}) in the $8$th iteration of the interior-point algorithm, in an LP with an {\it integral} solution. The constraint matrix $A$ has 830 rows and 3830 columns, $g_d = 48.6$, and $\kappa(Q) = 3.46 \times 10^4$.}
\label{PCGM_int1}
\end{figure}

\begin{figure}
\centering
\includegraphics[width=4.5 in] {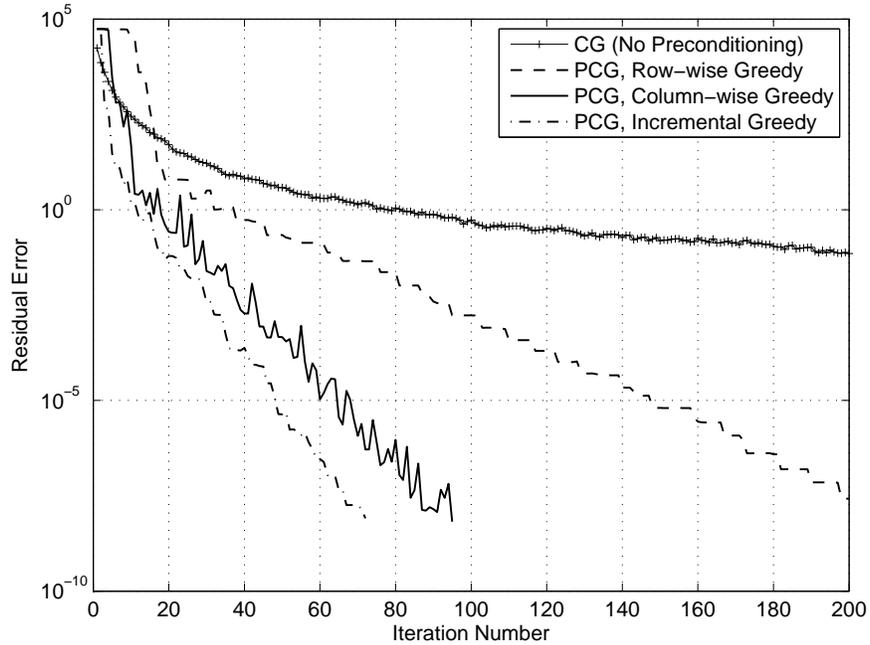}
\caption{The progress of the residual error for different PCG implementations, solving (\ref{Delta y}) in the $8$th iteration of the interior-point algorithm, in an LP with an {\it integral} solution. The constraint matrix $A$ has 830 rows and 3830 columns, $g_d = 0.22$, and $\kappa(Q) = 2.33 \times 10^8$.}
\label{PCGM_int2}
\end{figure}

To study the performance of the algorithms when the LP solution is not integral, we considered an LP from the 6th iteration of an MALP-A decoding problem at SNR = 1.0 dB, where the solution was \emph{fractional}. The matrix $A$ had 830 rows and 3830 columns. Fig. \ref{PCGM_frac1} corresponds to the $8$th iteration of the interior-point algorithm, with $g_d = 46.4$ and $\kappa(Q) = 2.03 \times 10^4$, while Fig. \ref{PCGM_frac2} corresponds to the $18$th (penultimate) iteration, with $g_d = 0.155$ and $\kappa(Q) = 2.61 \times 10^8$. These parameters are chosen such that the scenarios in these two figures are respectively similar to those in Figs. \ref{PCGM_int1} and \ref{PCGM_int2}, the main difference being that the decoding problem now has a fractional solution. We can observe that, while the performance of the CG method is very similar in Fig. \ref{PCGM_int1} and Fig. \ref{PCGM_frac1}, as well as in Fig.~\ref{PCGM_int2} and Fig.~\ref{PCGM_frac2}, the preconditioned implementations have slower convergence when the LP solution is fractional. In particular, in Fig. \ref{PCGM_frac2}, the row-wise greedy preconditioner does not improve the convergence of the CG method, and is essentially ineffective.

\begin{figure}
\centering
\includegraphics[width=4.5 in] {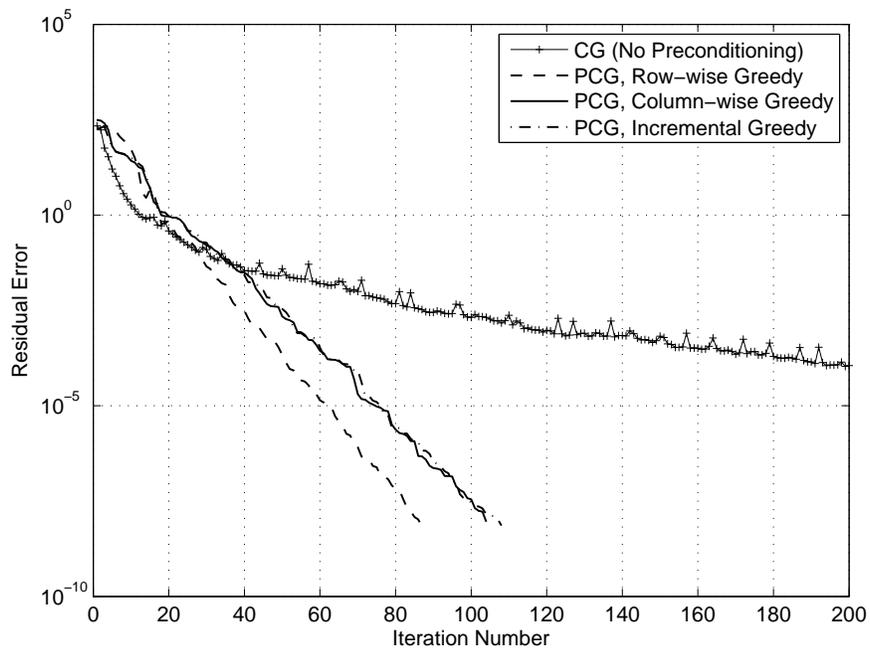}
\caption{The progress of the residual error for different PCG implementations, solving (\ref{Delta y}) in the $8$th iteration of the interior-point algorithm, in an LP with a {\it fractional} solution. The constraint matrix $A$ has 830 rows and 3830 columns, $g_d = 46.4$, and $\kappa(Q) = 2.03 \times 10^4$.}
\label{PCGM_frac1}
\end{figure}

\begin{figure}
\centering
\includegraphics[width=4.5 in] {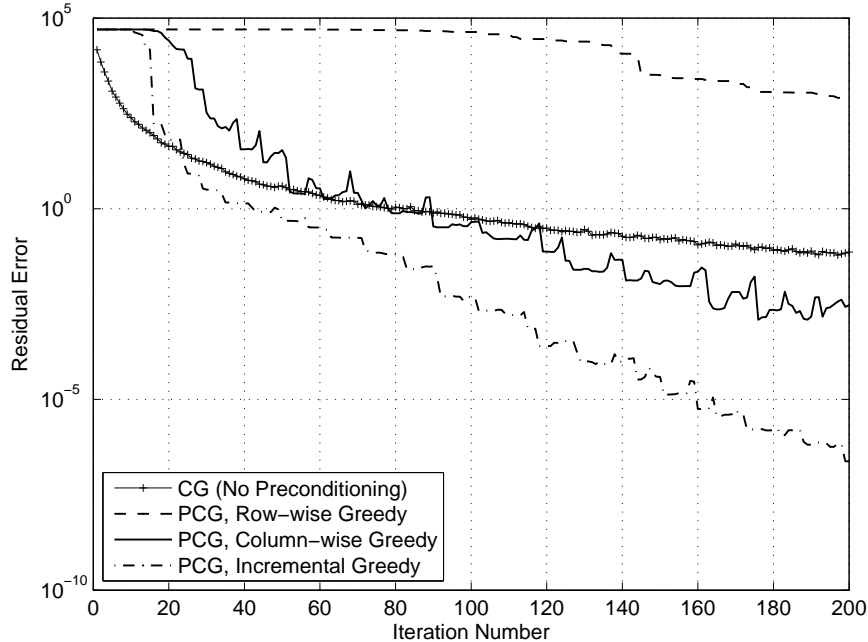}
\caption{The progress of the residual error for different PCG implementations, solving (\ref{Delta y}) in the $8$th iteration of the interior-point algorithm, in an LP with a {\it fractional} solution. The constraint matrix $A$ has 830 rows and 3830 columns, $g_d = 0.155$, and $\kappa(Q) = 2.61 \times 10^8$.}
\label{PCGM_frac2}
\end{figure}

\subsection{Discussion}
Overall, we have observed that in very ill-conditioned problems, the incremental and the column-wise greedy algorithms are significantly more effective than the row-wise greedy algorithm in speeding up the solution of the linear system. The better performance of the column-wise approach relative to the row-wise approach can be explained by the fact that the former, which searches for degree-1 columns, has more choices at each stage, since the columns of $A$ have lower degrees on average than its rows. Besides, while the column-wise is always able to find a complete triangular preconditioning set, the row-wise algorithm needs to expand the preconditioning set at the end by adding some slack columns that may have very low weights. Considering both the complexity and performance, the column-wise search (Algorithm~\ref{MTS col}) seems to be a suitable choice for a practical implemetation of LP decoding.

A second observation that we have made in our simulations is that the convergence of the PCG method cannot be well characterized just by the condition number of the preconditioned matrix. In fact, we have encountered several situations where the preconditioned matrix had a much higher condition number than the original matrix, yet it resulted in a much faster convergence. For instance, in the scenario studied in Fig. \ref{PCGM_frac2}, the condition number of the preconditioned matrix $M^{-1} Q$ for both the column-wise and the incremental algorithms was higher than that of $Q$ by factor of 50--100, while these preconditioners still improved the convergence compared to the CG method. Indeed, it is believed in the literature that the speed of convergence of the CG can typically be better explained by the number of distinct clusters of eigenvalues.

While we studied the interior-point method in the context of MALP decoding, the proposed algorithms can also be applied to the LPs that may have more than one constraint from each check node. For instance, we have observed that the proposed implementation is also very effective for ALP decoding. However, in the absence of the single-constraint property, some of the analytical results we presented may no longer be valid.

\section{Conclusion}
In this paper, we studied various elements in an efficient implementation of LP decoding. We first studied the adaptive LP decoding algorithm and two variations and demonstrated a number of properties of these algorithms. Specifically, we proposed modifications of the ALP decoding algorithm that satisfy the single-constraint property; i.e., each LP to be solved contains at most one parity inequality from each check node of the Tanner graph. 

We later studied a sparse interior-point implementation of linear programming, with the goal of exploiting the properties of the decoding problem in order to achieve lower complexity. The heart of the interior-point algorithm is the computation of the Newton step via solving an (often ill-conditioned) system of linear equations. Since iterative algorithms for solving sparse linear systems, including the conjugate-gradient method, converge slowly when the system is ill-conditioned, we focused on finding a suitable preconditioner to speed up the process. 

Motivated by the properties of LP decoding, we studied a new framework for desiging a preconditioner. Our approach was based on finding a square submatrix of the LP constraint matrix which contains the columns with the highest possible weights, and at the same time, can be made lower- or upper-triangular by column and row permutations, making it invertible in linear time. We proposed a number of greedy algorithms for designing such preconditioners, and proved that, when the solution to the LP is integral, two of these algorithms indeed result in effective preconditioners. We demonstrated the performance of the proposed schemes via simulation, and we observed that the preconditioned systems are most effective when the current LP has an integral solution. 

One can imagine various modifications and alternatives to the proposed greedy algorithms for designing preconditioners. It is also interesting to investigate the possibility of finding other adaptive or nonadaptive formulations of LP decoding that result in solving the fewest/smallest possible number of LPs, while maintaining the single-constraint property.
Moreover, there are several aspects of the implementation of LP decoding that are not explored in this work. These potential areas for future research include the optimum selection of the stopping criteria and step sizes for the interior-point algorithm and the CG method, as well as the theoretical analysis of the effect of preconditioning on the condition number and the eigenvalue spectrum of the linear system, similar to the study done in \cite{Judice} for network flow problems.
\appendices
\section{Proof of Theorem \ref{properties of ALP}}
\label{Proof of theorem}
\begin{proof} 
\begin{enumerate}
\renewcommand{\labelenumi}{\alph{enumi})}
\item To prove the claim, we show that the solution to any linear program $LP^k$ consisting of the $n$ initial (single-sided) box inequalities given by (\ref{initial constraints}) and any number of parity inequalities of the form (\ref{constraints2}) satisfies all the double-sided box constraints of the form $0 \leq u_i \leq 1,\ i \in \mathcal{I}=\{1,\ldots,n\}$.

For simplicity, we first transform each variable $u_i,\ i \in \mathcal{I}$, and its coefficient $\gamma_i$ in the objective function, respectively, into a new variable $v_i$ and a new coefficient $\lambda_i$, where 
\begin{equation}
\label{xform the LP}
\setlength{\nulldelimiterspace}{0pt}
\left\{\begin{IEEEeqnarraybox}[\relax][c]{l's}
v_i=u_i \hspace{0.26in} \text{ and } \lambda_i=\gamma_i \hspace{0.3in} \text{if}\ \ \gamma_i\geq0, \\
v_i=1-u_i \text{ and } \lambda_i=-\gamma_i \hspace{0.2in} \text{if}\ \ \gamma_i<0.
\end{IEEEeqnarraybox}\right.
\end{equation}
By this change of variables, we can rewrite $LP^k$ in terms of $v$. In this equivalent LP, all the variables $v_i$ will have nonnegative coefficients $\lambda_i$ in the objective function, and the box constraints (\ref{initial constraints}) will all be transformed into inequalities of the form $v_i \geq 0$. However, the transformed parity inequalities will still have the form
\begin{equation}
\label{xformed constraints}
\sum_{i\in \mathcal A_j} (1-v_i) + \sum_{i\in \mathcal B_j} v_i \geq 1,
\end{equation}
although here some of the sets $\mathcal A_j$ may have even cardinality. To prove the claim, it suffices to show that the unique solution $v^k$ to this LP satisfies $v^k_i\leq 1,\ \forall\ i\in \mathcal I$.

Assume, on the contrary, that for a subset of indices $\mathcal L \subseteq \mathcal I$, we have $v^k_i>1,\ \forall\ i\in \mathcal L$, and $0\leq v^k\leq 1,\ \forall\ i\in \mathcal I \backslash \mathcal L$. We define a new vector $\tilde v^k$ as
\begin{equation}
\label{tilde v}
\setlength{\nulldelimiterspace}{0pt}
\left\{\begin{IEEEeqnarraybox}[\relax][c]{l's}
\tilde v^k_i=1 \hspace{0.28in} \text{if}\ \ i\in \mathcal L, \\
\tilde v^k_i=v^k_i \hspace{0.2in} \text{if}\ \ i\in \mathcal I \backslash \mathcal L.
\end{IEEEeqnarraybox}\right.
\end{equation}
Remembering that $\lambda_i\geq 0,\ \forall\ i\in\mathcal I$, we will have $\lambda^T \tilde v^k \leq \lambda^T v^k$. Moreover, $\tilde v^k$ clearly satisfies all the double-sided box constraints $0\leq \tilde v^k_i \leq 1,\ \forall\ i\in \mathcal I$. 
We claim that any parity inequality of the form (\ref{xformed constraints}) in the LP, which is by assumption satisfied at $v^k$, is also satisfied at $\tilde v^k$. To see this, note that the first sum in (\ref{xformed constraints}) can only either increase or remain constant by moving from $v^k$ to $\tilde v^k$, and it will be nonnegative at $\tilde v^k$. Moreover, the second sum will remain constant if $\mathcal L \cap \mathcal B_j = \emptyset$, or will decrease but remain greater than or equal to one if $\mathcal L \cap \mathcal B_j \neq \emptyset$. In both cases, inequality (\ref{xformed constraints}) will be satisfied at $\tilde v^k$. Hence, we have shown that there is a feasible point $\tilde v^k$ which has a cost smaller than or equal to that of $v^k$. This contradicts the assumption that $v^k$ is the unique solution to the LP. Consequently, the solution to the LP should satisfy all the double-sided box constraints.

\item We need to show that $\gamma^T u^k < \gamma^T u^{k+1}$ for any $0\leq k< K$. This is obvious for ALP decoding, as the feasible set of $LP^k$ contains the feasible set of $LP^{k+1}$. For MALP-A and MALP-B, let ${LP}^{*k}$ be the problem obtained by removing from $LP^k$ a subset (or all) of the parity inequalities that are inactive at its solution, $u^k$. As discussed earlier, these inactive inequalities are non-binding, so the solution to ${LP}^{*k}$ must be $u^k$, as well. Now, $LP^{k+1}$ is obtained by adding some new (violated) constraints to ${LP}^{*k}$. Hence, the feasible set of ${LP}^{*k}$ strictly contains that of $LP^{k+1}$, which yields $\gamma^T u^k < \gamma^T u^{k+1}$.

\item Similar to the proof of \cite[Theorem 2]{ALP ISIT}. 

\item Similar to part b), let $LP^{*k}$ be the LP problem obtained by removing from $LP^k$ \emph{all} of the parity inequalities that are inactive at $u^k$, and remember that $u^k$ is the solution to $LP^{*k}$, as well. Clearly, all the parity inequalities in $LP^{*k}$ are from check nodes with indices in $\mathcal{J}^k$, thus the feasible space of $LP^{*k}$ contains that of $LPD^k$. Hence, it remains to show that $u^k$, the optimum feasible point for $LP^{*k}$, is also in the  feasible space of $LPD^k$.

Let $I^k \subseteq \{1, \ldots, n\}$ be the set of indices of variable nodes that are involved in at least one of the parity inequalities in $LP^{*k}$ (or, equivalently, check nodes in $\mathcal{J}^k$), and let $\tilde{I}^k$ be the set of the remaining indices. According to Corollary~\ref{one cut cor}, all the parity inequalities from check nodes in $\mathcal{J}^k$ are satisfied at $u^k$.  In addition, we can conclude from Corollary~\ref{redundant box} that the box constraints for variables with indices in $I^k$ are satisfied, as well.

Now, for any $i\in \tilde{I}^k$, the variable $u_i$ will be decoupled from all other variables, since it is only constrained by a box constraint according to (\ref{initial constraints}). Hence, in the solution $u^k$, such a variable will take the value $u^k_i=0$ if $\gamma_i > 0$ or $u^k_i=1$ if $\gamma_i<0$.\footnote{We assume that $\gamma_i \neq 0$, since otherwise, $u^k_i$ will not have a unique optimum value, which contradicts the uniqueness assumption on $u^k$ in the theorem.} Consequently,  $u^k$ satisfies all the parity inequalities and box constraints of $LPD^k$, and hence is the solution to this LP decoding problem.

\end{enumerate}
\end{proof}

\section*{Acknowledgment}
This work is supported in part by NSF Grant CCF-0829865.

\end{document}